\documentclass[aps,prl,groupedaddress,superscriptaddress,twocolumn,notitlepage,nofootinbib,bibnotes]{revtex4-2}
\pdfoutput=1

\let\coloneqq\relax
\let\eqqcolon\relax

\usepackage[latin1]{inputenc}
\usepackage{amsthm}
\usepackage{amssymb}
\usepackage{amsmath}
\usepackage{bbold}
\usepackage{bbm}
\usepackage[pagebackref=true, colorlinks=true, allcolors=MidnightBlue!70!black!85!TealBlue]{hyperref}
\usepackage{braket}
\usepackage{dsfont}
\usepackage{mathdots}
\usepackage{mathtools}
\usepackage{enumerate}
\usepackage[shortlabels]{enumitem}
\usepackage{csquotes}
\usepackage{stmaryrd}
\usepackage[cal=boondox]{mathalfa}
\usepackage{graphicx}
\usepackage{stackengine}
\usepackage{scalerel}
\usepackage{tensor}       
\usepackage[dvipsnames, table]{xcolor}
\usepackage{array}
\usepackage{makecell}
\newcolumntype{x}[1]{>{\centering\arraybackslash}p{#1}}
\usepackage{tikz}
\usepackage{pgfplots}
\usetikzlibrary{shapes.geometric, shapes.misc, positioning, arrows, arrows.meta, decorations.pathreplacing, decorations.pathmorphing, patterns, angles, quotes, calc}
\usepackage{booktabs}
\usepackage{xfrac}
\usepackage{siunitx}
\usepackage{centernot}
\usepackage{comment}
\usepackage{chngcntr}
\usepackage{caption}
\usepackage{subcaption}
\usepackage{ragged2e}

\usepackage{newpxtext,newpxmath}

\newcommand{\ra}[1]{\renewcommand{\arraystretch}{#1}}

\newtheorem{thm}{Theorem}
\newtheorem*{thm*}{Theorem}
\newtheorem{prop}[thm]{Proposition}
\newtheorem*{prop*}{Proposition}
\newtheorem{lemma}[thm]{Lemma}
\newtheorem*{lemma*}{Lemma}
\newtheorem{cor}[thm]{Corollary}
\newtheorem*{cor*}{Corollary}

\newtheorem*{cj*}{Conjecture}

\newtheorem*{Def*}{Definition}

\newtheorem*{question*}{Question}

\newtheorem*{problem*}{Problem}

\makeatletter
\def\thmhead@plain#1#2#3{%
  \thmname{#1}\thmnumber{\@ifnotempty{#1}{ }\@upn{#2}}%
  \thmnote{ {\the\thm@notefont#3}}}
\let\thmhead\thmhead@plain
\makeatother

\theoremstyle{definition}
\newtheorem{rem}[thm]{Remark}
\newtheorem*{note}{Note}

\newtheorem{manualthminner}{Theorem}
\newenvironment{manualthm}[1]{%
  \renewcommand\themanualthminner{#1}%
  \manualthminner \it
}{\endmanualthminner}

\newcommand{\bb}{\begin{equation}\begin{aligned}\hspace{0pt}}
\newcommand{\bbb}{\begin{equation*}\begin{aligned}}
\newcommand{\ee}{\end{aligned}\end{equation}}
\newcommand{\eee}{\end{aligned}\end{equation*}}
\newcommand\floor[1]{\left\lfloor#1\right\rfloor}
\newcommand\ceil[1]{\left\lceil#1\right\rceil}
\newcommand{\eqt}[1]{\stackrel{\mathclap{\scriptsize \mbox{#1}}}{=}}
\newcommand{\leqt}[1]{\stackrel{\mathclap{\scriptsize \mbox{#1}}}{\leq}}

\newcommand{\geqt}[1]{\stackrel{\mathclap{\scriptsize \mbox{#1}}}{\geq}}
\newcommand{\ketbra}[1]{\ket{#1}\!\!\bra{#1}}

\newcommand{\sumno}{\sum\nolimits}

\newcommand{\e}{\varepsilon}
\renewcommand{\epsilon}{\varepsilon}

\newcommand{\dd}{\mathrm{d}}

\newcommand{\id}{\mathds{1}}

\newcommand{\N}{\mathds{N}}

\newcommand{\C}{\mathds{C}}

\newcommand{\ve}{\varepsilon}

\newcommand{\cptp}{\mathrm{CPTP}}

\newcommand{\locc}{\mathrm{LOCC}}
\newcommand{\sep}{\mathrm{SEP}}
\newcommand{\ppt}{\mathrm{PPT}}

\newcommand{\SEP}{\pazocal{S}}
\newcommand{\PPT}{\pazocal{P\!P\!T}}
\newcommand{\sepp}{\mathrm{NE}}

\DeclareMathOperator{\Tr}{Tr}

\DeclareMathOperator{\co}{conv}
\DeclareMathOperator{\cone}{cone}

\DeclareMathAlphabet{\pazocal}{OMS}{zplm}{m}{n}

\DeclareMathOperator{\supp}{supp}

\newcommand{\HH}{\pazocal{H}}

\newcommand{\NN}{\pazocal{N}}

\newcommand{\MM}{\pazocal{M}}

\newcommand{\XX}{\pazocal{X}}
\newcommand{\TT}{\pazocal{T}}

\newcommand{\FF}{\pazocal{F}}

\newcommand{\lsmatrix}{\left(\begin{smallmatrix}}
\newcommand{\rsmatrix}{\end{smallmatrix}\right)}

\newcommand{\deff}[1]{\textbf{\emph{#1}}}

\stackMath
\newcommand\xxrightarrow[2][]{\mathrel{%
  \setbox2=\hbox{\stackon{\scriptstyle#1}{\scriptstyle#2}}%
  \stackunder[5pt]{%
    \xrightarrow{\makebox[\dimexpr\wd2\relax]{$\scriptstyle#2$}}%
  }{%
   \scriptstyle#1\,%
  }%
}}

\newcommand{\tends}[2]{\xxrightarrow[\! #2 \!]{\mathrm{#1}}}
\newcommand{\tendsn}[1]{\xxrightarrow[\! n\rightarrow \infty\!]{#1}}

\stackMath

\makeatletter
\newcommand*\rel@kern[1]{\kern#1\dimexpr\macc@kerna}
\newcommand*\widebar[1]{%
  \begingroup
  \def\mathaccent##1##2{%
    \rel@kern{0.8}%
    \overline{\rel@kern{-0.8}\macc@nucleus\rel@kern{0.2}}%
    \rel@kern{-0.2}%
  }%
  \macc@depth\@ne
  \let\math@bgroup\@empty \let\math@egroup\macc@set@skewchar
  \mathsurround\z@ \frozen@everymath{\mathgroup\macc@group\relax}%
  \macc@set@skewchar\relax
  \let\mathaccentV\macc@nested@a
  \macc@nested@a\relax111{#1}%
  \endgroup
}

\newcommand{\fakepart}[1]{
 \par\refstepcounter{part}
  \sectionmark{#1}
}

\counterwithin*{equation}{part}
\counterwithin*{thm}{part}
\counterwithin*{figure}{part}

\tikzset{meter/.append style={draw, inner sep=10, rectangle, font=\vphantom{A}, minimum width=30, line width=.8, path picture={\draw[black] ([shift={(.1,.3)}]path picture bounding box.south west) to[bend left=50] ([shift={(-.1,.3)}]path picture bounding box.south east);\draw[black,-latex] ([shift={(0,.1)}]path picture bounding box.south) -- ([shift={(.3,-.1)}]path picture bounding box.north);}}}
\tikzset{roundnode/.append style={circle, draw=black, fill=gray!20, thick, minimum size=10mm}}
\tikzset{squarenode/.style={rectangle, draw=black, fill=none, thick, minimum size=10mm}}

\definecolor{Blues5seq1}{RGB}{239,243,255}
\definecolor{Blues5seq2}{RGB}{189,215,231}
\definecolor{Blues5seq3}{RGB}{107,174,214}
\definecolor{Blues5seq4}{RGB}{49,130,189}
\definecolor{Blues5seq5}{RGB}{8,81,156}

\definecolor{Greens5seq1}{RGB}{237,248,233}
\definecolor{Greens5seq2}{RGB}{186,228,179}
\definecolor{Greens5seq3}{RGB}{116,196,118}
\definecolor{Greens5seq4}{RGB}{49,163,84}
\definecolor{Greens5seq5}{RGB}{0,109,44}

\definecolor{Reds5seq1}{RGB}{254,229,217}
\definecolor{Reds5seq2}{RGB}{252,174,145}
\definecolor{Reds5seq3}{RGB}{251,106,74}
\definecolor{Reds5seq4}{RGB}{222,45,38}
\definecolor{Reds5seq5}{RGB}{165,15,21}

\allowdisplaybreaks

\newcommand{\dne}{\mathrm{DNE}}
\newcommand{\dpptp}{\mathrm{DPPTP}}
\newcommand{\ane}{\mathrm{ANE}}
\newcommand{\adne}{\mathrm{ADNE}}
\newcommand{\ff}{\mathds{M}}
\newcommand{\mm}{\mathds{M}}
\renewcommand{\sep}{\mathds{SEP}}
\renewcommand{\ppt}{\mathds{PPT}}
\newcommand{\mlocc}{\mathds{LOCC}}
\newcommand{\all}{\mathds{ALL}}
\newcommand{\pptop}{\mathrm{PPT}}
\newcommand{\kk}{\widetilde{E}_\kappa}
\newcommand{\stein}{\mathrm{Stein}}

\usepackage[left=.73in,right=.73in,top=.7in,bottom=1in]{geometry}

\begin{document}

\title{Distillable entanglement under dually non-entangling operations}

\author{Ludovico Lami}
\email{ludovico.lami@gmail.com}
\affiliation{QuSoft, Science Park 123, 1098 XG Amsterdam, the Netherlands}
\affiliation{Korteweg--de Vries Institute for Mathematics, University of Amsterdam, Science Park 105-107, 1098 XG Amsterdam, the Netherlands}
\affiliation{Institute for Theoretical Physics, University of Amsterdam, Science Park 904, 1098 XH Amsterdam, the Netherlands}

\author{Bartosz Regula}
\email{bartosz.regula@gmail.com}
\affiliation{Mathematical Quantum Information RIKEN Hakubi Research Team, RIKEN Cluster for Pioneering Research (CPR) and RIKEN Center for Quantum Computing (RQC), Wako, Saitama 351-0198, Japan}

\begin{abstract}
Computing the exact rate at which entanglement can be distilled from noisy quantum states is one of the longest-standing questions in quantum information. We give an exact solution for entanglement distillation under the set of dually non-entangling (DNE) operations --- a relaxation of the typically considered local operations and classical communication, comprising all channels which preserve the sets of separable states and measurements. We show that the DNE distillable entanglement coincides with a modified version of the regularised relative entropy of entanglement in which the arguments are measured with a separable measurement. Ours is only the second known regularised formula for the distillable entanglement under any class of free operations in entanglement theory, after that given by Devetak and Winter for (one-way) local operations and classical communication. An immediate consequence of our finding is that, under DNE, entanglement can be distilled from any entangled state. As our second main result, we construct a general upper bound on the DNE distillable entanglement, using which we prove that the separably measured relative entropy of entanglement can be strictly smaller than the regularisation of the standard relative entropy of entanglement, solving an open problem posed by Li and Winter. 
Finally, we study also the reverse task of entanglement dilution and show that the restriction to DNE operations does not change the entanglement cost when compared with the larger class of non-entangling operations. This implies a strong form of irreversiblility of entanglement theory under DNE operations: even when asymptotically vanishing amounts of entanglement may be generated, entangled states cannot be converted reversibly.
\end{abstract}

\maketitle

\let\oldaddcontentsline\addcontentsline
\renewcommand{\addcontentsline}[3]{}
\fakepart{Main text}

\section{Introduction}

Entanglement distillation is poised to be one of the fundamental primitives of the emerging field of quantum technologies~\cite{Bennett-distillation, Bennett-distillation-mixed, Bennett-error-correction}. Its goal is to transform many copies of a noisy entangled state $\rho_{AB}$ emitted by some source into (a smaller number of) \emph{ebits}, i.e.\ perfect units of entanglement, that can then be used as fuel in a variety of quantum protocols, e.g.\ teleportation~\cite{teleportation}, dense coding~\cite{dense-coding}, or quantum key distribution~\cite{Ekert91, RennerPhD}, as well as to demonstrate violations of Bell inequalities~\cite{Brunner-review}. To perform this transformation, one is allowed to use quantum operations taken from a class of \emph{free operations} $\FF$, assumed to comprise all quantum channels that are easy to implement on our quantum devices with the resources at our disposal. 
The fundamental figure of merit that characterises the ultimate efficiency of entanglement distillation on $\rho_{AB}$ is the \emph{distillable entanglement} under $\FF$, denoted by $E_{d,\FF}(\rho_{AB})$. This depends critically both on the input state $\rho_{AB}$ 
and on the set of free operations $\FF$ with which distillation is carried out.

Historically, the first class of free operations to be widely studied was that of \emph{local operations and classical communication} (LOCC). LOCCs are the most general operations that can be performed by spatially separated parties that are bound to the rules of quantum mechanics and can communicate only classically~\cite{LOCC}. In spite of their operational importance, LOCCs are exceedingly difficult to characterise mathematically, and as a consequence we currently do not possess any formula 
to compute the LOCC distillable entanglement on arbitrary states, or even decide when it is zero and when it is not~\cite{Horodecki-open-problems, list-open-problems}. This is to be contrasted with the fact that we \emph{do} possess a formula for the LOCC entanglement \emph{cost}, which is given by the regularised entanglement of formation~\cite{Hayden-EC}.

To bypass this problem, in recent years there has been an increasing interest in characterising entanglement manipulation under classes of free operations that 
approximate LOCCs. There are at least two valid reasons to pursue this goal, besides the obvious one that simpler classes of free operations offer an ideal test-bed to improve our understanding of entanglement manipulation. 
First, in this way one can obtain lower and upper bounds on the LOCC distillable entanglement, because if $\FF_1 \subseteq \locc \subseteq \FF_2$ then $E_{d,\,\FF_1}(\rho_{AB}) \leq E_{d,\,\locc}(\rho_{AB}) \leq E_{d,\,\FF_2}(\rho_{AB})$. Several of the best known bounds
on the distillable entanglement can be seen as stemming from this approach. For example, the restriction to one-way LOCCs yields the hashing lower bound~\cite[Theorem~13]{devetak2005}, while going to the larger set of PPT operations results in the upper bounds via the 
negativity~\cite{negativity, plenioprl} and 
the Rains bound~\cite{Rains2001}. Nevertheless, even under simpler classes of operations, the ultimate capabilities of entanglement manipulation --- and in particular the exact value of distillable entanglement --- are not known, with the 
exception of one-way LOCC~\cite{devetak2005}.

But there is also another, more fundamental reason to go beyond the LOCC paradigm. Since the early days of entanglement theory, a whole line of research~\cite{Popescu1997, vedral_2002, Horodecki2002} has focused on the conceptually striking parallels between the laws of entanglement manipulation, which govern the interconversion of pure and mixed entanglement, and those of thermodynamics, which govern the interconversion of work and heat. In building a `thermodynamic theory of entanglement', a top-down, axiomatic approach is somewhat preferable to the bottom-up approach that leads to the definition of LOCC, and arguably closer in spirit to the original founding principles of thermodynamics~\cite{CARNOT, Clausius, Kelvin}. This approach has been already quite fruitful, shedding light on the fundamental question of (ir)reversibility of entanglement manipulation~\cite{BrandaoPlenio1, BrandaoPlenio2,irreversibility,gap}, but many of 
the issues it raises remain unresolved.

In this work, we focus on a particular superset of LOCC, namely the \emph{dually non-entangling} (DNE) operations. DNE operations, originally introduced by Chitambar et al.~\cite{Chitambar2020}, can be easily understood from a rather natural axiomatic perspective as those that preserve separability (i.e.\ absence of entanglement) of states when seen in the Schr\"odinger picture \emph{and} of measurements when seen in the Heisenberg picture. An important point is that DNE protocols are more restricted than some other commonly employed operations such as non-entangling (NE) maps~\cite{BrandaoPlenio1,BrandaoPlenio2}, which means that they provide a closer approximation to LOCCs and can potentially yield improved bounds on their operational processing power.
Here we completely characterise asymptotic entanglement manipulation under DNE operations,
connecting it with the important operational task of entanglement testing, and revealing new features of phenomena such as entanglement irreversibility and bound entanglement.

\section{Results}

\textbf{\em Our contribution.}--- Our first main result is a clean formula for the DNE distillable entanglement (Theorem~\ref{distillable_dne_thm}), which can be expressed as a regularised separably measured divergence $D^{\sep,\infty}(\cdot\|\SEP)$ ---
a modified version of the relative entropy of entanglement $D^{\infty}(\cdot\|\SEP)$~\cite{Vedral1997, Vedral1998} that was originally introduced by Piani in a seminal work~\cite{Piani2009} and has recently found some applications~\cite{faithful, rel-ent-sq, Berta2023,gap}. We dub it the Piani relative entropy of entanglement.
To compute it, one calculates the distance of a given state from the set of separable states as measured by a 
form of measured relative entropy 
that involves an optimisation over separable measurements only. In spite of the seemingly more complicated definition, the resulting quantity enjoys a plethora of desirable properties, some of which, 
such as strong super-additivity~\cite[Theorem~2]{Piani2009}, are often very useful~\cite{Piani2009,brandao_adversarial,catboundent} but do not hold for the standard relative entropy of entanglement~\cite{Werner-symmetry}. As a by-product of our result, many of these properties are inherited by the DNE distillable entanglement.
Perhaps the most important one is the \emph{faithfulness} --- since the Piani relative entropy is non-zero for all entangled states, this immediately implies that the phenomenon of bound entanglement does not exist under DNE operations.

Our result 
represents one of the few instances where the distillable entanglement under some 
relevant class of free operations can be calculated, albeit in terms of some regularised expression. \emph{Regularisation} here means that one needs to understand the limiting behaviour of the given quantity when acting on many copies of quantum states. The necessity for such a procedure is a consequence of the fact that operational quantities encountered in quantum information are typically not additive~\cite{Werner-symmetry, Shor2004, Hayden-p>1, Cubitt-p-->0, Hastings2008, superactivation}, precluding the validity of non-regularised (`single-letter') expressions. 
We should remark here that while regularised formulas do not allow for a straightforward calculation because of the 
ubiquity of non-additivity phenomena 
in quantum information, they are nevertheless amenable to theoretical investigation, which is the reason why many cornerstone results in quantum information theory, such as the Lloyd--Shor--Devetak theorem~\cite{Lloyd-S-D, L-Shor-D, L-S-Devetak}, or the Holevo--Schumacher--Westmoreland theorem~\cite{Holevo-S-W, H-Schumacher-Westmoreland}, are proofs of regularised formulas. 
Besides our result, the only other known formula for distillable entanglement of \emph{any} kind is the result of Devetak and Winter~\cite{devetak2005}, where a regularised expression is found for the rate of distillation under local operations assisted by one- or two-way classical communication; however, the expressions there are not expressed in the form of simple entropic quantities and the complex optimisation problems involved mean that their computability and applicability is questionable.
A regularised expression has also been conjectured for the distillable entanglement under the class of non-entangling operations~\cite{BrandaoPlenio2,gap} in terms of $D^\infty(\cdot\|\SEP)$ --- this conjecture being equivalent to the notorious generalised quantum Stein's lemma~\cite{Brandao2010,gap-comment}.

To complement our study of entanglement distillation, we also provide a comprehensive characterisation of the reverse task of entanglement dilution under DNE operations. The relevant quantity here is the entanglement cost $E_{c,\FF}(\rho_{AB})$, that is, the rate at which maximally entangled states are needed in order to produce a given noisy state.
We show that the entanglement cost under DNE operations is the same as the one under the larger class of non-entangling operations. On the one hand, this means that it suffices to employ the smaller and more restrictive 
set of DNE operations to achieve optimal dilution rates. On the other hand, the result allows us to use a number of previously established findings in the study of NE operations~\cite{BrandaoPlenio2,irreversibility} to directly constrain or even exactly compute the DNE entanglement cost. We will see this to have crucial consequences.

One of the most important application of the axiomatic approaches to entanglement theory has been the study of entanglement reversibility~\cite{BrandaoPlenio1,BrandaoPlenio2,irreversibility,gap}, namely the question of whether the rate of entanglement distillation $E_{d,\FF}(\rho_{AB})$ equals the entanglement cost $E_{c,\FF}(\rho_{AB})$. The fundamental consequence of such reversibility would be the establishment of a `second law' of entanglement: asymptotic entanglement transformations would be completely governed by a single function, playing a role analogous to entropy in thermodynamics.
So far, reversibility has not been shown under any class of quantum channels~\cite{gap}, but there is a strong contender that has been conjectured to yield a reversible entanglement theory: namely, asymptotically non-entangling (ANE) maps~\cite{BrandaoPlenio1,BrandaoPlenio2}. Although entanglement is known to be irreversible under non-entangling channels~\cite{irreversibility}, the entanglement manipulation rates can be increased by allowing small amounts of entanglement (according to some appropriate entanglement measure) to be generated in the process, as long as they vanish asymptotically.
Using our characterisation of entanglement dilution under DNE maps, we can show that the entanglement cost under DNE maps enhanced by such asymptotic entanglement generation (which we may refer to as asymptotically dually non-entangling maps, ADNE) equals the corresponding entanglement cost under ANE maps (Theorem~\ref{cost_dne_thm}). But the latter is already known~\cite{BrandaoPlenio2}: it is given exactly by the regularised relative entropy of entanglement $D^{\infty}(\cdot\|\SEP)$.


We thus obtain a complete understanding of the 
theory of entanglement manipulation under asymptotically dually non-entangling channels: the distillable entanglement is given by the regularised Piani relative entropy $D^{\sep,\infty}(\cdot\|\SEP)$, while the entanglement cost by the standard relative entropy of entanglement $D^{\infty}(\cdot\|\SEP)$.
The question of reversibility of entanglement under ADNE operations is thus reduced to the question of whether there exists a gap between the quantities $D^{\infty}(\cdot\|\SEP)$ and $D^{\sep,\infty}(\cdot\|\SEP)$.
A very closely related question was previously asked by Li and Winter~\cite[Figure~1]{rel-ent-sq}. 
With our third main result (Theorem~\ref{antisymmetric_gap_thm}) we solve 
both of these problems, proving that the antisymmetric state $\alpha_d$ satisfies 
$D^{\sep,\infty}(\alpha_d\|\SEP) < D^\infty(\alpha_d\|\SEP)$ for large enough $d$, namely $d \geq 13$.
Our results thus directly prove that entanglement manipulation is irreversible under DNE operations, even if asymptotic entanglement generation is allowed.
Such a strong form of irreversibility contrasts with the conjectured reversible framework under ANE operations~\cite{BrandaoPlenio2,gap} and shows that the restriction to DNE operations carries strong physical consequences.
Our result also solves the open problem mentioned in~\cite{rel-ent-sq} and further reinforces the fame of the antisymmetric state as the `universal counterexample' in quantum entanglement theory~\cite{Aaronson2008}.

\begin{table*}[h!t] \centering
\ra{1.7} \setlength{\tabcolsep}{1pt}
\begin{tabular}{cccc} \toprule
Operations & $\begin{array}{c} \text{Distillable} \\[-8pt] \text{entanglement} \end{array}$ & $\begin{array}{c} \text{Entanglement} \\[-8pt] \text{cost} \end{array}$ & Reversibility? \\
\midrule
NE & $D^\infty(\cdot \| \SEP)\ \ \raisebox{-.4pt}{\text{\large \color{Reds5seq5} \textbf{?}}}$ & {\large \color{Purple} $\star$} & No \cite{irreversibility} \\
ANE & 
$D^\infty(\cdot \| \SEP)\ \ \raisebox{-.4pt}{\text{\large \color{Reds5seq5} \textbf{?}}}$ & $D^\infty(\cdot \| \SEP)$ & Yes $\raisebox{-.4pt}{\text{\large \color{Reds5seq5} \textbf{?}}}$\ \ \cite{BrandaoPlenio2, gap} \\
DNE & \cellcolor{Greens5seq2!45!white} \hspace{1ex} $D^{\sep,\infty}(\cdot\|\SEP)$ \hspace{1ex} & \cellcolor{Greens5seq2!45!white} {\large \color{Purple} $\star$} & No \cite{irreversibility}\\
ADNE & \cellcolor{Greens5seq2!45!white} $D^{\sep,\infty}(\cdot\|\SEP)$ & \cellcolor{Greens5seq2!45!white} $D^\infty(\cdot\|\SEP)$ & \cellcolor{Greens5seq2!45!white} No ({\small Thms.~\ref{distillable_dne_thm}--\ref{antisymmetric_gap_thm}}) \\
\bottomrule
\end{tabular}
\caption[LoF entry]{\justifying
The distillable entanglement $E_{d,\FF}$ and the entanglement cost $E_{c,\FF}$ under four different classes of free operations $\FF$, namely, non-entangling (NE), asymptotically non-entangling (ANE), dually non-entangling (DNE), and asymptotically dually non-entangling (ADNE) operations. The white cells summarise results in the prior literature, while the green ones contain our new findings. The distillable entanglement under NE and ANE is the same, and it is conjectured to coincide with the regularised relative entropy of entanglement $D^\infty(\cdot\|\SEP)$
~\cite{BrandaoPlenio2,Brandao2010,gap,gap-comment}. 

\vspace{3pt}
\hspace{9pt} In this work, we completely characterise entanglement manipulation under DNE and ADNE channels, which are less permissive and thus more practically relevant than their NE or ANE counterparts. For both classes, the distillable entanglement is given by the regularised Piani relative entropy $D^{\sep,\infty}(\cdot\|\SEP)$ (Theorem~\ref{distillable_dne_thm}). The ADNE entanglement cost is the same as that under ANE, and it thus equals $D^\infty(\cdot\|\SEP)$ (Theorem~\ref{cost_dne_thm}). The cost under DNE, instead, coincides with that under the NE class: while it is natural to conjecture that this is also given by \emph{some} regularised relative entropy expression, we do not yet have any explicit expression for it, and we thus marked it with the symbol $\color{Purple} \star$. 
Our irreversibility result in Theorem~\ref{antisymmetric_gap_thm} then implies that there is a gap between the distillable entanglement and entanglement cost under (A)DNE operations; furthermore, through a careful analysis of asymptotic hypothesis testing, we can show an operational gap in the power of NE and DNE operations in entanglement distillation.
}
\label{rates_table}
\end{table*}

\medskip
\textbf{\em Notation.}--- A quantum state $\rho_{AB}$ on a finite-dimensional bipartite quantum system $AB$ is called \deff{separable} if it can be written as $\rho_{AB} = \sum_x p_x\, \alpha_x^A\otimes \beta_x^B$, and it is called \deff{entangled} otherwise. We denote the set of separable states on $AB$ by $\SEP(A\!:\!B)$, or simply $\SEP$ if there is no ambiguity concerning the underlying system. We will also be interested in the set of \deff{separable measurements}, denoted by $\sep(A\!:\!B)$ or simply $\sep$. It comprises all positive operator-valued measures (POVM) $(E_x)_x$ such that $E_x\in \cone(\SEP)$ for all $x$, where $\cone(\SEP)$ denotes the cone of un-normalised separable operators. Notably, all LOCC measurements are separable, but the converse is not true~\cite{LOCC}. 

We can quantify the entanglement of a state $\rho_{AB}$ by calculating its distance from the set of separable states as measured by some quantum divergence. The two main choices of quantum divergence we will be concerned with here are the \deff{quantum (Umegaki) relative entropy}~\cite{Umegaki1962}, given by $D(\rho\|\sigma) \coloneqq \Tr\left[ \rho \left(\log_2\rho - \log_2\sigma\right) \right]$, and the \deff{Piani relative entropy}, which is defined for two arbitrary bipartite states $\rho=\rho_{AB}$ and $\sigma=\sigma_{AB}$ by~\cite{Piani2009}
\bb
D^{\sep} (\rho\|\sigma) \coloneqq \sup_{(E_x)_x\,\in\, \sep} \sum_x \Tr[\rho E_x] \log_2 \frac{\Tr [E_x \rho]}{\Tr [E_x \sigma]}\, .
\ee
Yet another possible choice is the \deff{max-relative entropy}, defined by $D_{\max}(\rho\|\sigma)\coloneqq \min\{ \lambda:\ \rho\leq 2^\lambda \sigma\}$. The \deff{relative entropy of entanglement}, the \deff{Piani relative entropy of entanglement}, and the \deff{max-relative entropy of entanglement} are then defined by~\cite{Vedral1997, Vedral1998, Piani2009}
\bb
\mathds{D}(\rho\|\SEP) \coloneqq \min_{\sigma\in \SEP} \mathds{D}(\rho\|\sigma)\, ,\qquad \mathds{D}= D,\, D^{\sep}\!,\, D_{\max}\, .
\ee
In quantum information one often looks at the asymptotic limit of many copies, which captures the ultimate limitations to which quantum phenomena are subjected, and is reminiscent of the thermodynamic limit in statistical physics. Doing so yields the \deff{regularised} entanglement measures
\bb
\mathds{D}^\infty(\rho\|\SEP) \coloneqq \lim_{n\to\infty} \frac1n\, \mathds{D}\big(\rho^{\otimes n} \,\big\|\, \SEP_n\big)\, ,\quad \mathds{D}= D,\, D^{\sep}\!,\, D_{\max}\, .
\ee
where $\SEP_n = \SEP(A^n\!:\!B^n)$. Because of the hierarchical relations existing between the corresponding divergences, it holds that
\bb
D^{\sep}(\rho\|\SEP) \leq D(\rho\|\SEP) \leq D_{\max}(\rho\|\SEP)\, ,
\ee
and analogously for the regularised quantities. These three entanglement measures have however remarkably different features. While $D(\rho\|\SEP)$ and  $D_{\max}(\rho\|\SEP)$ are sub-additive on several independent systems, $D^{\sep}(\rho\|\SEP)$ is not only super-additive but actually \deff{strongly super-additive}, meaning that~\cite[Theorem~2]{Piani2009}
\bb
&D^\sep\big(\rho_{AA'BB'}\, \big\|\, \SEP(AA'\!:\!BB')\big) \\
&\qquad \geq D^\sep\big(\rho_{AB}\, \big\|\, \SEP(A\!:\!B)\big) + D^\sep\big(\rho_{A'B'}\, \big\|\, \SEP(A'\!:\!B')\big)
\label{main_strong_superadditivity}
\ee
for all (possibly correlated) four-partite states $\rho_{AA':BB'}$.

For a given class of \deff{free operations} $\FF$ and some bipartite state $\rho_{AB}$, the corresponding \deff{distillable entanglement} $E_{d,\,\FF}(\rho_{AB})$ can be defined as the maximum number of ebits that can be extracted per copy of $\rho_{AB}$, in the asymptotic limit where many copies are available, using operations from $\FF$ while making a vanishingly small error. 
Although $\FF=\locc$ is the traditional choice~\cite{Bennett-distillation-mixed}, for the aforementioned reasons we are interested here in the larger class of dually non-entangling operations, previously introduced in~\cite{Chitambar2020}. This can be justified axiomatically as follows. Any free operation $\Lambda = \Lambda_{AB\to A'B'}$ that defines a valid entanglement manipulation protocol should transform separable states into separable states, i.e.\ it should not inject additional entanglement into the system. Operations that satisfy this constraint are called \deff{non-entangling} (NE)~\cite{BrandaoPlenio1, irreversibility}. While this appropriately captures what should happen in the Schr\"odinger picture where $\Lambda$ acts on states, we can also look at the Heisenberg picture, where $\Lambda^\dag$, the adjoint defined by $\Tr X \Lambda(Y) = \Tr \Lambda^\dag(X) Y$, acts on measurement operators. If we also assume that pre-processing by $\Lambda$ should not turn a separable measurement into a non-separable one, we obtain the set of \deff{dually non-entangling} (DNE) operations, defined by the conditions~\cite{Chitambar2020}
\bb
\Lambda(\sigma_{AB}) &\in \SEP(A'\!\!:\!B') &\forall\ \sigma_{AB} &\in \SEP(A\!:\!B)\, , \\
\Lambda^\dag (E_{A'B'}) &\in \cone\!\big(\SEP(A\!:\!B)\big)\ \ &\forall\ E_{A'B'} &\in \cone\!\big(\SEP(A'\!\!:\!B')\big)\, .
\label{DNE}
\ee
Generalising the concept of asymptotically non-entangling operations (ANE) found in the framework of Brand\~{a}o and Plenio~\cite{BrandaoPlenio1,BrandaoPlenio2}, 
we can also consider \deff{asymptotically DNE} (ADNE) operations, in which we allow 
the creation of 
some entanglement, as long as this amount is sufficiently small: it must vanish asymptotically
according to some fixed entanglement measure. Following~\cite{BrandaoPlenio1,BrandaoPlenio2}, 
we quantify the generated entanglement by the max-relative entropy of entanglement, replacing the first line of~\eqref{DNE} by the condition that $\max_{\sigma\in \SEP} D_{\max}\left(\Lambda_n(\sigma) \|\SEP\right) \leq \eta_n$ for the operation $\Lambda_n$ used for entanglement distillation at the $n$-copy level, and then requiring that 
$\eta_n\tendsn{} 0$. We will denote the corresponding distillable entanglement by $E_{d,\, \adne}(\rho_{AB})$.

\medskip
\textbf{\em Distillable entanglement.}--- We will now state our first main result, which establishes a regularised formula for the distillable entanglement under (A)DNE operations.

\begin{thm} \label{distillable_dne_thm}
For all bipartite states $\rho=\rho_{AB}$, the distillable entanglement under (asymptotically) dually non-entangling operations coincides with the regularised Piani relative entropy of entanglement, in formula
\bb
E_{d,\,\dne}(\rho) = E_{d,\,\adne}(\rho) = D^{\sep,\infty}(\rho\|\SEP)\, .
\label{distillable_dne_and_adne}
\ee
As a consequence, there is no bound entanglement under DNE: for any entangled state $\rho$, it holds that
\bb
E_{d,\,\dne}(\rho) \geq D^{\sep}(\rho\|\SEP) > 0\, .
\label{quantitative_faithfulness_distillable_dne}
\ee
\end{thm}

Before we sketch the main proof ideas, several remarks are in order.

(I)~Eq.~\eqref{quantitative_faithfulness_distillable_dne} provides a faithful single-letter lower bound on the DNE distillable entanglement of any state, and in particular, thanks 
to the faithfulness of the Piani relative entropy~\cite{Piani2009}, it reveals that there is no bound entanglement under DNE operations.
Distillability under these operations was previously studied in~\cite[Theorem~13]{Chitambar2020}, although a complete proof that all entangled states are distillable was not obtained there. Theorem~\ref{distillable_dne_thm} also provides a conceptually cleaner proof of our previous result~\cite[Proposition~7]{gap}, where we showed that there is no bound entanglement under ANE operations; it suffices to observe that since ADNE operations are in particular ANE, we have that $E_{d,\, \ane}(\rho) \geq E_{d,\, \adne}(\rho) = D^{\sep,\infty}(\rho\|\SEP)$.

(II)~The 
result of Eq.~\eqref{distillable_dne_and_adne} shows that the choice between DNE or ADNE operations makes no difference as far as distillation is concerned. This is a priori unexpected, as ADNE operations are strictly more powerful than DNE in general, but not entirely surprising, as the same happens for non-entangling vs.\ asymptotically non-entangling operations~\cite[Lemma~S17]{irreversibility}. Furthermore, in 
Supplementary Note~\ref{subsec_ADNE_distillable} we prove that even allowing ADNE operations to generate an arbitrary sub-linear amount of entanglement (i.e.\ $\eta_n = o(n)$) does not increase the distillation rate.

(III)~Explicitly characterising various features of distillable entanglement, such as its additivity properties, is often extremely difficult. However, equating the rate of distillation with an entropic quantity helps mitigate such problems. Thanks to Theorem~\ref{distillable_dne_thm}, the (A)DNE distillable entanglement is seen to inherit all the useful properties of the Piani relative entropy of entanglement. A notable example is the
strong super-additivity~\eqref{main_strong_superadditivity} --- this feature of the (A)DNE distillable entanglement, 
highly non-trivial to see from its definition, establishes a strong parallel with the LOCC distillable entanglement, which is also known to be strongly super-additive~\cite[Table~3.4]{MatthiasPhD}.

(IV) The presence of regularisation in the formula~\eqref{distillable_dne_and_adne} may make the quantity look dauntingly difficult to compute in practice. However, there are several ways in which it may be efficiently estimated. First, the super-additivity of the Piani relative contrasts with the \emph{sub}-additivity of the standard relative entropy of entanglement, which means that obtaining lower bounds for the former is as easy as upper bounds for the latter: it suffices to evaluate it on a single copy of a quantum state, yielding $D(\rho\|\SEP) \geq E_{d,\,\dne}(\rho) \geq D^{\sep}(\rho\|\SEP)$.
Furthermore, 
we will shortly introduce a general single-letter bound, which will often allow for an efficient evaluation of $D^{\sep,\infty}(\rho\|\SEP)$  in practice.

\medskip
\textbf{\em Proof sketch of Theorem~\ref{distillable_dne_thm} and connection with hypothesis testing}. --- Our approach will be to first establish a tight bound on the \emph{one-shot} DNE distillable entanglement, that is, on the amount of entanglement that can be extracted from a single copy of a state $\rho$. By 
studying the asymptotic behaviour of this quantity when many copies of $\rho$ are available, we will then obtain a formula for the rate of distillation.

To this end, we will relate DNE entanglement distillation with a seemingly different task, namely, 
hypothesis testing of entangled states --- or \deff{entanglement testing} for short~\cite{Brandao2010, gap-comment}. 
In this latter task,
given an unknown quantum state which could be either the entangled state $\rho$ or any separable state, the goal is to perform a measurement and guess which of the two hypotheses is true. Our setting is different from that of~\cite{Brandao2010, gap-comment} in that we assume that the only allowed measurements are separable, i.e.\ they do not carry any entanglement themselves. 
In the 
asymmetric setting, one assumes that the probability of mistaking $\rho$ for a separable state is at most $\e$ and asks about the opposite error probability. The least achievable error 
can then be quantified by a quantity known as the separably measured hypothesis testing relative entropy~\cite{brandao_adversarial}, which is defined as 
\bb
D_H^{\sep\!,\,\e}(\rho\|\sigma) \coloneqq - \log_2 \inf_{(E,\,\id-E)\in\sep}\left\{ \Tr E\sigma\!:\, \Tr E\rho\!\geq\! 1\!-\!\e \right\} \!.
\ee

Returning to entanglement distillation, let us denote by $E_{d,\,\dne_\eta}^{(1),\,\e}\!(\rho)$ the number of ebits that can be distilled from a single copy of a bipartite state $\rho=\rho_{AB}$, up to error $\e$ and up to $\eta$ of generated entanglement. Some thoughtful manipulation --- the details of which we defer to Supplementary Note~\ref{sec_distillable} --- then shows that
\bb
\label{eq:oneshot_ent}
\floor{D_H^{\sep\!,\,\e}(\rho\| \SEP)} 
&\leq E_{d,\,\dne_\eta}^{(1),\,\e} 
\!(\rho) \\
&\leq D_H^{\sep\!,\,\e}(\rho\| \SEP) + \eta + 1.
\ee
This tells us that the DNE distillable entanglement is, up to some small terms that will asymptotically vanish, completely determined by $D_H^{\sep\!,\,\e}$. This connection between DNE entanglement distillation and entanglement testing with separable measurements is not only one of our main conceptual contributions, but on the more practical side it
allows us to characterise the properties of entanglement distillation by using results from hypothesis testing. 
And indeed, the asymptotic behaviour of $D_H^{\sep\!,\,\e}$ was already studied in Ref.~\cite{brandao_adversarial}, where it was shown that
\bb\label{eq:measured_stein}
\lim_{\e\to 0}\liminf_{n\to\infty} \frac1n D_H^{\sep\!,\,\e}\big(\rho^{\otimes n}\,\big\|\, \SEP_n\big) = D^{\sep,\infty}(\rho\|\SEP).
\ee
By Eq.~\eqref{eq:oneshot_ent}, the expression on the left-hand side is precisely the asymptotic (A)DNE distillable entanglement, from which Eq.~\eqref{distillable_dne_and_adne} follows. Eq.~\eqref{quantitative_faithfulness_distillable_dne} is then a consequence of known properties of $D^{\sep}$~\cite{Piani2009}.
This concludes the sketch of the proof of Theorem 1. The complete technical details of all of our proofs can be found in the Supplementary Information.

The idea of using entanglement testing to provide bounds for distillable entanglement has a long history~\cite{Vedral1998,HAYASHI}, but it was not until the work of Brand\~ao and Plenio~\cite{BrandaoPlenio2} that a precise equivalence between the two concepts was established in a suitable axiomatic setting. 
Specifically, they showed that the rate of distillation under non-entangling operations can be expressed using a related hypothesis testing problem as
\begin{equation}\begin{aligned}\label{eq:ne_distillable_stein}
    E_{d,\rm NE}(\rho) = \lim_{\e\to 0}\liminf_{n\to\infty} \frac1n D_H^{\all,\,\e}\big(\rho^{\otimes n}\,\big\|\, \SEP_n\big),
\end{aligned}\end{equation}
where the notation $\all$ refers to all measurements being allowed. 
There is, however, a crucial difference between Brand\~{a}o and Plenio's work~\cite{BrandaoPlenio2} and ours: while the former constructed a one-to-one mapping between non-entangling protocols and \emph{global} entanglement tests, the proof of our Theorem~\ref{distillable_dne_thm}, and in particular that of Eq.~\eqref{eq:oneshot_ent}, shows that DNE distillation protocols --- including those that are not implementable with LOCC --- are in one-to-one correspondence with 
\emph{separable} entanglement tests. 
Our result is significant because, in practice, separable measurements are easier to implement experimentally in a distant local laboratory setting in which Alice and Bob do not have access to unbounded quantum communication.

\medskip
\textbf{\em Entanglement cost.}---  It is also of interest to understand the task opposite to distillation, 
the dilution of entanglement --- that is, the use of ebits to produce a 
target (mixed) quantum state. The rate at which this can be done is known as the \emph{entanglement cost} $E_{c,\,\FF}(\rho_{AB})$, and the question of reversibility of entanglement manipulation~\cite{Vidal-irreversibility,Horodecki2002,BrandaoPlenio1,BrandaoPlenio2,irreversibility,gap} asks whether $E_{d,\,\FF}(\rho_{AB})=E_{c,\,\FF}(\rho_{AB})$ holds for all states under the given class of free operations $\FF$. 

Our second main result establishes an equivalence between (A)DNE and the larger class of (A)NE operations in entanglement dilution, showing that any dilution task that could be achieved with Brand\~ao and Plenio's (A)NE operations~\cite{BrandaoPlenio2,irreversibility} can also be accomplished with the more axiomatically meaningful (A)DNE operations.

\begin{thm} \label{cost_dne_thm}
For all bipartite states $\rho=\rho_{AB}$, the entanglement cost under (asymptotically) dually non-entangling operations coincides with the corresponding entanglement cost under (asymptotically) non-entangling operations:
\begin{align}
E_{c,\,\dne}(\rho) &= E_{c,\,\sepp}(\rho), \label{cost_dne} \\
E_{c,\,\adne}(\rho) &= E_{c,\,\ane}(\rho) = D^{\infty}(\rho\|\SEP)\, . \label{cost_adne}
\end{align}
\end{thm}
The last equality 
in~\eqref{cost_adne} comes from the known result~\cite{BrandaoPlenio2} that the ANE entanglement cost of a state $\rho$ is given by the regularised relative entropy of entanglement $D^\infty(
\rho \| \SEP)$. It was also recently shown that $E_{c,\,\sepp}(\rho) > D^\infty(\rho \| \SEP)$ in general~\cite{irreversibility}, so a similar equivalence cannot hold for DNE operations without asymptotic entanglement generation.

Theorems~\ref{distillable_dne_thm} and~\ref{cost_dne_thm} together show that the asymptotic properties of entanglement manipulation under ADNE operations are governed by two entropic quantities: Piani relative entropy $D^{\sep,\infty}(\cdot\|\SEP)$ in distillation, and the relative entropy $D^\infty(\cdot\|\SEP)$ in dilution.
This provides a direct motivation for the question: is there actually any difference between the two entropic quantities, or are they simply equal?

\medskip
\textbf{\em Gap with the relative entropy.}--- 
To demonstrate a gap between 
$D^{\sep,\infty}(\rho\|\SEP)$ and $D^\infty(\rho \| \SEP)$, we introduce a simple but general single-letter upper bound for the regularised Piani relative entropy: namely,
\bb
E_{d,\,\dne}(\rho) =  D^{\sep,\infty}(\rho \| \SEP)  \leq \kk^\infty(\rho) \leq \kk(\rho)\, ,
\label{upper_bound_DNE_distillable}
\ee
where
\bb
\label{main_kappa_def}
\kk(\rho) \coloneqq \log_2 \min\left\{ \Tr S:\ -S \leq \rho^\Gamma \leq S ,\ S\!\in\! \cone(\SEP) \right\}
\ee
is a modification of an entanglement monotone  $E_\kappa$ studied in a different context by Wang and Wilde~\cite{Wang2023}. As before, $\cone(\SEP)$ denotes here the cone of separable operators. Any ansatz for Eq.~\eqref{main_kappa_def} then yields an upper bound on $E_{d,\, \dne}(\rho)$, allowing one to estimate its value.
We now apply the bound of Eq.~\eqref{upper_bound_DNE_distillable} to the so-called \deff{antisymmetric state} --- a special entangled state that has been 
the subject of many investigations~\cite{Werner, Christandl2012, lancien2016} in light of its peculiar properties, earning 
the name of `universal counterexample' in entanglement theory~\cite{Aaronson2008}. It is proportional to the projector onto the antisymmetric subspace within the bipartite Hilbert space $\C^d\otimes \C^d$, thus being given by $\alpha_d\coloneqq \frac{\id-F}{d(d-1)}$, where $F\ket{\psi}\!\ket{\phi} \coloneqq \ket{\phi}\!\ket{\psi}$ is the swap operator. It is 
$(d-1)$-extendible 
on either sub-system, thus it has small distillable entanglement, distillable key, and squashed entanglement --- all $O(1/d)$~\cite{Christandl2012}. At the same time, its entanglement cost and regularised relative entropy of entanglement $D^\infty(\alpha_d\|\SEP)$ are larger than a fixed constant independent of $d$~\cite{Christandl2012}. Our third main result states that the regularised Piani relative entropy of entanglement $D^{\sep,\infty}(\alpha_d\|\SEP)$ behaves unlike $D^\infty(\alpha_d\|\SEP)$ and instead goes to $0$ as $d\to\infty$, again as $O(1/d)$. The  existence of a state with this property was left as an open problem by Li and Winter~\cite[Figure~1]{rel-ent-sq}.

\begin{thm} \label{antisymmetric_gap_thm}
The antisymmetric state $\alpha_d$ satisfies that
\bb
D^{\sep,\infty}(\alpha_d\|\SEP) \leq \log_2 \left(1 + \frac2d\right) \tends{}{d\to\infty} 0\, ,
\ee
while $D^\infty(\alpha_d\|\SEP) \geq \frac12 \log_2 \frac43 \approx 0.2075$ for all $d$. Consequently, $D^{\sep,\infty}(\alpha_d\|\SEP) < D^{\infty}(\alpha_d\|\SEP)$ for all $d \geq 13$.
\end{thm}

Although 
this might seem a purely mathematical result, we can give it a direct physical meaning in two contexts: one, in the manipulation of entanglement, and two, in quantum hypothesis testing. We will now explain both of these interpretations in detail.

\medskip
\textbf{\em Irreversibility of entanglement manipulation.}--- 
Theorems~\ref{distillable_dne_thm},~\ref{cost_dne_thm}, and~\ref{antisymmetric_gap_thm} together imply a fundamental irreversibility of entanglement under asymptotically dually non-entangling operations: it holds that
\bb
E_{c,\,\adne}(\alpha_d) > E_{d,\,\adne}(\alpha_d)
\ee
for the antisymmetric Werner state with $d \geq 13$. What this means is that there is no hope of establishing reversibility under DNE transformations, even assisted by asymptotic entanglement generation.

This further hints at the fact that ANE might truly be the smallest class of operations that enable reversibility of entanglement manipulation, as suggested by Plenio~\cite{OpenProblemArxiv}. The validity of this conjecture hinges on the generalised quantum Stein's lemma~\cite{Brandao2010, gap}, arguably one of the main open problems in the field. Our results shed further light on this problem, pinpointing which assumptions are key to reversibility.
In the recent work~\cite{irreversibility}, we showed that asymptotic entanglement generation is necessary: allowing only non-entangling operations leads to irreversibility.
Here we show that the entanglement generation in itself is not enough: one really needs the full operational power of non-entangling operations \emph{and} asymptotic entanglement generation in order to have any hope of establishing a reversible theory of entanglement.
This teaches us something not only about entanglement, but also about the mathematical structure of quantum resource theories: namely, that asymptotically dually resource non-generating operations~\cite{RT-review} are in general not enough to unlock reversibility.

Interestingly, the antisymmetric state $\alpha_d$ can actually be shown to be \emph{reversible} under some other classes of free operations that extend LOCC, such as so-called PPT operations~\cite{Martin-exact-PPT}. This shows that the more limited distillation power of (A)DNE operations has crucial operational consequences, and (A)DNE is a useful outer approximation to LOCC that is rather independent from other commonly employed ones such as PPT. 

We note also that recently a framework was formulated that enables reversibility for all quantum states by allowing for the use probabilistic ANE operations in addition to quantum channels~\cite{regula_2023}. Although very closely related to the question studied here, this requires a suitable adjustment of the definition of asymptotic rates, and it is not directly comparable with the standard definitions of transformation rates used in this work.

\medskip
\textbf{\em Entanglement testing with separable measurements.}--- Recall from Eq.~\eqref{eq:measured_stein} that the Piani relative entropy $D^{\sep,\infty}(\rho\|\SEP)$ exactly quantifies the asymptotic performance in hypothesis testing of against all separable states using only separable measurements.

Already in the works~\cite{Brandao2010,BrandaoPlenio2} an important conjecture was made: that when all measurements are allowed --- not just separable ones, but also general global measurements --- then the asymptotic rate of entanglement testing equals the regularised relative entropy of entanglement $D^\infty(\rho\|\SEP)$. Because of the connection between entanglement testing and entanglement distillation under NE operations (Eq.~\eqref{eq:ne_distillable_stein}), this would also show that $E_{d,\rm NE}(\rho) =D^\infty(\rho\|\SEP)$ and hence that entanglement is reversible under NE. 
This conjecture is precisely the generalised quantum Stein's lemma. We have already shown a gap between $D^{\sep,\infty}(\rho\|\SEP)$ and $D^{\infty}(\rho\|\SEP)$, but without the generalised quantum Stein's lemma, it is impossible to conclude if this gap extends to the operational performance of NE and DNE or to entanglement testing.

Fortunately, we can show the conjectured generalised quantum Stein's lemma of~\cite{Brandao2010} holds true for the state $\alpha_d$, implying that
$E_{d,\rm NE}(\alpha_d) = D^\infty(\alpha_d\|\SEP)$ for this state. (See Supplementary Note~\ref{sec:antisym_note}.) 
Together with Theorems~\ref{distillable_dne_thm} and~\ref{antisymmetric_gap_thm}, this then shows that
\begin{equation}\begin{aligned}
    E_{d,\rm DNE}(\alpha_d) < E_{d,\rm NE}(\alpha_d) 
\end{aligned}\end{equation}
for sufficiently large $d$. On the one hand, this exhibits a gap in the operational performance of the two sets of operations in entanglement distillation, showing that DNE provide a strictly tighter approximation to LOCC.
On the other hand, this directly shows a gap in the performance of hypothesis testing in the two settings: in the entanglement testing of the antisymmetric state $\alpha_d$, global measurements can do \emph{strictly better} than separable (and hence also local) ones. This is the first known asymptotic gap of this kind.

\section{Discussion}

We have contributed to the asymptotic theory of entanglement manipulation by exactly computing the rate at which entanglement can be distilled from any quantum state using dually non-entangling operations. In addition to being one of the few examples where asymptotic rates for general quantum states can be evaluated in terms of a single regularised quantity, 
this result also gives a direct operational interpretation to the Piani relative entropy of entanglement, a known mathematical tool that had not been directly connected with entanglement manipulation before. 
We provided new insights into the asymptotic properties of the antisymmetric state --- 
a state that has received significant attention due to its often unusual entanglement properties, and 
whose characterisation is an important problem in 
entanglement theory --- by resolving an open 
question regarding its distinguishability. Finally, 
we have shown that dually non-entangling operations assisted by asymptotically vanishing amounts of entanglement can dilute entanglement at a rate given by the regularised relative entropy of entanglement, matching the performance of the larger class of asymptotically non-entangling operations and ruling out the reversibility of entanglement in this setting.

Despite the axiomatic character of DNE operations, our results lead to the establishment of new practically relevant connections --- in particular, that between DNE entanglement distillation and entanglement testing with separable measurements --- and our precise characterisation of the properties of DNE sheds light on several important physical phenomena in entanglement theory. 
For instance, through the hypothesis testing connections built and strengthened here, we showed that there is a strict gap in the asymptotic performance of separable vs.\ global measurements in distinguishing entangled states from unentangled ones. Further,
we now know that DNE operations exhibit no bound entanglement, and this immediately implies that a solution to the important open problem of the existence of NPT bound entanglement may only be obtained by looking at classes of operations strictly smaller than DNE. Similarly, our irreversibility result shows that any reversible framework must use operations larger than ADNE.
Such insights can be regarded as no-go results that illuminate the extremely complicated and still little-understood questions about the power of different operations in transforming entangled states.

We hope that our results can find use in the understanding of the often enigmatic landscape of asymptotic entanglement manipulation. Characterising the distillability and reversibility properties, as well as evaluating the distillable entanglement for other types of operations remain major open problems in the field.

\medskip
\textbf{\em Note added.}\,--- At the time of publication, two proofs of the generalised quantum Stein's lemma have been claimed~\cite{Hayashi-Stein, GQSL}. To recall, the validity of this result would establish the conjectured equality between the NE distillable entanglement and the regularised relative entropy of entanglement $D^\infty(\cdot\|\SEP)$ for all states, complementing our main results. Our investigation of the antisymmetric state in Supplementary Note~\ref{sec:antisym_note} provides an independent and simpler proof of a special case of this lemma.

We also note that the result of Theorem~\ref{antisymmetric_gap_thm} can be deduced from a claim found in the first pre-print version of~\cite{Cheng2020} (Table~1, $4^\text{th}$ row, $3^\text{rd}$ column), available at \href{https://arxiv.org/abs/2011.13063v1}{https://arxiv.org/abs/2011.13063v1}. No complete proof of this claim was available when the first pre-print of this work appeared, in July~2023. A proof using a different technique than ours can however be found in the published version of~\cite{Cheng2020} (Corollary~14).

\medskip
\textbf{\em Acknowledgements.}\,--- LL thanks the Freie Universit\"at Berlin for hospitality. LL and BR are grateful to Julio I.\ de~Vicente, 
Eric Chitambar, Marco Tomamichel, and Andreas Winter for useful comments on the manuscript. We also thank Mario Berta for insightful discussions.

\bibliographystyle{apsc}
\bibliography{biblio}
\let\addcontentsline\oldaddcontentsline

\clearpage
\newgeometry{left=1.2in,right=1.2in,top=.7in,bottom=1in}

\fakepart{Appendix}

\onecolumngrid

\begin{center}
\vspace*{\baselineskip}
{\textbf{\large Distillable entanglement under dually non-entangling operations \\[8pt] --- Supplementary information --- }}\\
\end{center}

\renewcommand{\theequation}{S\arabic{equation}}
\renewcommand{\thethm}{S\arabic{thm}}
\renewcommand{\thefigure}{S\arabic{figure}}
\setcounter{page}{1}
\makeatletter

\setcounter{secnumdepth}{2}

\tableofcontents

\section{Definitions}

\subsection{Generalities}

\subsubsection{States and operations}

Quantum states are represented by \deff{density operators}, i.e.\ positive semi-definite operators of unit trace acting on a separable Hilbert space. Unless otherwise specified, in this paper we will only consider finite-dimensional Hilbert spaces. A \deff{measurement operator} on a system with Hilbert space $\HH$ is an operator $E$ on $\HH$ satisfying $0\leq E\leq \id$, where $X\geq Y$ means that $X-Y$ is positive semi-definite. We will also write the above relation as $E\in [0,\id]$. 

A \deff{quantum channel} in dimension $d$ is a linear map $\NN$ acting on the space of $d\times d$ complex matrices that is completely positive and trace preserving (CPTP). Given a channel $\NN$, the dual map $\NN^\dag$, alias the \deff{adjoint}, is defined by the formula
\bb
\Tr \big[ X \NN(Y)\big] = \Tr \big[ \NN^\dag(X) Y\big]\, ,\qquad \forall\ X,\, Y\, ,
\label{dual_map}
\ee
where $X,Y$ are arbitrary $d\times d$ complex matrices. We will denote the set of quantum channels having system $A$ as input and system $B$ as output by $\cptp(A\to B)$. A special type of quantum channels are quantum \deff{measurements}, a.k.a.\ positive operator-valued measures (POVMs), denoted by $\MM$ and given by
\bb
\MM(\cdot) \coloneqq \sum_x \Tr[E_x (\cdot)] \ketbra{x}\, ,
\ee
where $x$ runs over a finite alphabet $\XX$, $\{\ket{x}\}_x$ is the canonical orthonormal basis of the Hilbert space $\C^{|\XX|}$, and $E_x\in [0,\id]$ is a measurement operator for every $x$. In order for the above map to be also trace preserving, and hence a quantum channel, the normalisation condition $\sum_x E_x =\id$ should be obeyed.

\subsubsection{Relative entropies} 

Umegaki's \deff{relative entropy} is a generalisation of the Kullback--Leibler divergence~\cite{Kullback-Leibler} to quantum systems. For a pair of finite-dimensional states $\rho,\sigma$ on the same system, it is defined by~\cite{Umegaki1962, Hiai1991}
\bb
D(\rho\|\sigma) \coloneqq \Tr[\rho (\log_2 \rho - \log_2\sigma)] ,
\label{Umegaki}
\ee
where it is understood that $D(\rho\|\sigma) = +\infty$ unless $\supp(\rho)\subseteq \supp(\sigma)$. This is by no means the only possible generalisation of the classical quantity to the quantum world. An alternative expression based on the idea of measuring the quantum states to obtain classical probability distributions, on which we can then evaluate the Kullback--Leibler divergence, has been proposed by Donald~\cite{Donald1986, Petz-old, Berta2017, nonclassicality}. More generally, one can consider a set of `allowed' measurements $\mm$ on a certain system, and define the \deff{$\mm$-measured relative entropy}~\cite{Piani2009} (see also~\cite{Vedral1998}) to be
\bb
D^{\mm} (\rho\|\sigma) \coloneqq \sup_{(E_x)_x\in \mm} \sum_x \Tr \rho E_x \log_2 \frac{\Tr \rho E_x}{\Tr \sigma E_x}\, .
\label{measured_relative_entropy}
\ee

Another related quantity is the \deff{max-relative entropy}, given by~\cite{Datta08}
\bb
D_{\max}(X\|\sigma) \coloneqq \log_2 \min\left\{\lambda\geq 0:\ X \leq \lambda \sigma\right\} ,
\label{max_relative_entropy}
\ee
where by convention $\log_2 0 = -\infty$. Although $D_{\max}$ is usually defined for the case where $X$ and $\sigma$ are both quantum states,~\eqref{max_relative_entropy} actually makes perfect sense when $X,\sigma$ are arbitrary Hermitian operators. We will find this slightly more general definition quite handy. Note that $D_{\max}(X\|\sigma)\geq 0$ provided that $\Tr X\geq 1$.

The Umegaki relative entropy and the max-relative entropy both obey the \deff{data processing inequality}. In the former case, this reads~\cite{lieb73a, lieb73b, Lindblad-monotonicity}
\bb
D\big(\NN(\rho) \,\big\|\, \NN(\sigma) \big) \leq D(\rho \| \sigma)\, ,
\label{data_processing_Umegaki}
\ee
which is valid for all quantum channels $\NN$ and for all pairs of states $\rho,\sigma$. In fact,~\eqref{data_processing_Umegaki} holds provided that $\NN$ is positive and trace preserving, so that \emph{complete} positivity is not even necessary~\cite{Alex2017}. Similarly, the max-relative entropy satisfies that
\bb
D_{\max}\big(\NN(X) \,\big\|\, \NN(\sigma) \big) \leq D_{\max} (X \| \sigma)\, ,
\label{data_processing_D_max}
\ee
for all positive maps $\NN$, and in particular for all quantum channels. An analogous inequality for the measured relative entropy also holds, provided that a simple compatibility condition is obeyed. Namely, for all sets of measurements $\mm,\mm'$ and for all channels $\NN$ such that $\NN^\dag(\mm)\subseteq \mm'$, it holds that
\bb
D^\mm\big(\NN(\rho)\,\big\|\, \NN(\sigma) \big) \leq D^{\mm'}(\rho\|\sigma)\, .
\label{data_processing_measured_relent}
\ee
Finally, there is a simple hierarchical relation between the three relative entropies discussed here. Namely, for all states $\rho,\sigma$ and all sets of measurements $\mm$ it holds that
\bb
D^{\mm}(\rho\|\sigma) \leq D(\rho \|\sigma) \leq D_{\max}(\rho\|\sigma)\, ;
\label{Umegaki_max_inequality}
\ee
here, the first inequality follows from~\eqref{data_processing_Umegaki}, and the second from the operator monotonicity of the logarithm.

\subsubsection{Entanglement and PPT-ness}

Consider an arbitrary finite-dimensional bipartite quantum system $AB$ with Hilbert space $\HH_{AB} = \HH_A\otimes \HH_B$. A state $\sigma_{AB}$ on $AB$ is called \deff{separable}~\cite{Werner} if there exist $N\in \N_+$, a probability distribution $p$ on $\{1,\ldots, N\}$, and $N$ pairs of states $\alpha^i_A$ on $A$ and $\beta^i_B$ on $B$ ($i=1,\ldots, N$) such that
\bb
\sigma_{AB} = \sum_{i=1}^N p_i\, \alpha^i_A \otimes \beta^i_B\, .
\label{separable}
\ee
A bipartite state that does not admit any decomposition of the above form is called \deff{entangled}. We will denote the set of separable states on $AB$ by $\SEP(A\!:\!B)$, or simply $\SEP$ if there is no ambiguity regarding the system. The cone generated by $\SEP$ will be denoted by $\cone(\SEP)\coloneqq \{\lambda \sigma:\ \lambda\geq 0,\ \sigma\in \SEP\}$.

An incredibly useful outer approximation of the set of separable states is provided by the \deff{positive partial transpose (PPT) criterion}~\cite{PeresPPT}, stating that any separable state $\sigma_{AB}$ satisfies the operator inequality $\sigma_{AB}^\Gamma \geq 0$, where the \deff{partial transpose} $\Gamma$ of an operator acting on the finite-dimensional space $\HH_{AB}$ is defined by 
\bb
\Gamma \left( X_A\otimes Y_B\right) = \left( X_A\otimes Y_B\right)^\Gamma \coloneqq X_A \otimes Y_B^\intercal
\ee
on product operators, and extended by linearity. A state $\sigma_{AB}$ satisfying the PPT criterion is called a \deff{PPT state}. The set of PPT states is denoted by $\PPT(A\!:\!B)$, or simply by $\PPT$ when there is no ambiguity concerning the system under consideration. Although $\SEP \subseteq \PPT$, the inclusion is known to be strict whenever the bipartite system has dimension larger than $6$~\cite{Horodecki-PPT-entangled}.


The above definitions, valid for states, can be easily extended to measurements. A measurement represented by the POVM $(E_x)_x$ is said to be an \deff{LOCC measurement} if it can be implemented with local operations and classical communication. This notion can be turned into a rigorous mathematical definition as detailed in~\cite{LOCC}. Outer approximations to the set of LOCC measurements, denoted by $\mlocc$, are given by the sets of \deff{separable measurements} and \deff{PPT measurements}, respectively defined by
\bb
\sep(A\!:\!B) &\coloneqq \left\{ \big(E^{AB}_x\big)_x:\ E^{AB}_x\in \SEP(A\!:\!B)\ \ \forall\ x,\ \sumno_x E_x^{AB} = \id^{AB} \right\} , \\
\ppt(A\!:\!B) &\coloneqq \left\{ \big(E^{AB}_x\big)_x:\ E^{AB}_x\in \PPT(A\!:\!B)\ \ \forall\ x,\ \sumno_x E_x^{AB} = \id^{AB} \right\} .
\ee
As mentioned, it holds that
\bb
\mlocc(A\!:\!B) \subseteq \sep(A\!:\!B) \subseteq \ppt(A\!:\!B) \subseteq \all(A\!:\!B) \, ,
\label{measurement_inclusions}
\ee
where $\all \coloneqq \left\{(E_x)_x:\ E_x\geq 0,\ \sumno_x E_x = \id\right\}$ is the set of all measurements. Typically, all inclusions in~\eqref{measurement_inclusions} are strict~\cite{LOCC, Chitambar2014, Cheng2020}.

\begin{note}
Throughout this paper, we will leave the operationally important $\mlocc$ set aside, and consider only its outer approximations $\sep$ and $\ppt$, as well as, naturally, $\all$. Therefore, from now on, we will use the symbol $\mm$ for one of the three latter sets of measurements, i.e.\ $\mm\in \{\sep,\ppt,\all\}$ unless otherwise specified.
\end{note}

\subsection{Entanglement measures} \label{subsec_entanglement_measures}

An entanglement measure is a function defined on states of all (finite-dimensional) bipartite systems and obeying certain properties like monotonicity under LOCC~\cite{Horodecki-review}. 
A general family of entanglement measures, sometimes called `entanglement divergences', can be constructed by setting
\bb
\mathds{D}(\rho\|\SEP) \coloneqq \inf_{\sigma\in \SEP} \mathds{D}(\rho\|\sigma)\, ,
\label{divergence_of_entanglement}
\ee
where $\mathds{D}\in \{D,\, D_{\max},\, D^\mm\}$ is some `quantum divergence'. Roughly speaking,~\eqref{divergence_of_entanglement} can be thought of as quantifying the `distance' between the state $\rho$ and the set of separable states, as measured by the given quantum divergence. Here the quotation marks are appropriate, because divergences are not symmetric in their two arguments, and hence they do not represent, properly speaking, distance measures.

Choosing $\mathds{D} = D$ in~\eqref{divergence_of_entanglement} yields the celebrated \deff{relative entropy of entanglement}~\cite{Vedral1997, Vedral1998}; setting $\mathds{D} = D_{\max}$, instead, yields~\cite{Datta-alias}
\bb
D_{\max}(\rho\|\SEP) = \log_2\left(1+R_\SEP^g(\rho)\right)\, ,
\label{robustness_as_D_max}
\ee
where $R_\SEP^g$, known as the \deff{generalised robustness}~\cite{VidalTarrach, Steiner2003}, is given by
\bb
R_\SEP^g(\rho) &\coloneqq \min\left\{ r\geq 0:\ \exists\ \sigma\in \SEP:\ \rho \leq (1+r)\sigma \right\} .
\label{generalised_robustness}
\ee
Finally, picking $\mathds{D} = D^\sep$ in~\eqref{divergence_of_entanglement} yields a quantity introduced by Piani~\cite{Piani2009}, which we will therefore call \deff{Piani relative entropy of entanglement}.

These entanglement measures can be defined also in infinite dimension, and especially the relative entropy of entanglement and the generalised robustness are typically surprisingly well behaved~\cite{achievability, taming-PRA, taming-PRL}.

The three measures above have different properties: when $\mathds{D} = D,D_{\max}$ the corresponding entanglement divergence is sub-additive, i.e.
\bb
\mathds{D}\big(\rho_{AB} \otimes \rho'_{A'B'}\, \big\|\, \SEP(AA'\!:\!BB')\big) \leq \mathds{D}\big(\rho_{AB}\, \big\|\, \SEP(A\!:\!B)\big) + \mathds{D}\big(\rho'_{A'B'}\, \big\|\, \SEP(A'\!:\!B')\big)\, ,\qquad \mathds{D} = D,\, D_{\max}\, ,
\label{subadditivity}
\ee
for all pairs of bipartite states $\rho_{AB},\rho'_{A'B'}$. This means that combining independent systems can only decrease the total amount of entanglement as measured by either $D(\cdot\|\SEP)$ or $D_{\max}(\cdot\|\SEP)$. However, a breakthrough result by Piani~\cite[Theorem~2]{Piani2009} showed that when $\mm=\sep$ the corresponding measured relative entropy of entanglement is \deff{strongly super-additive}, i.e.\
\bb
D^\sep\big(\rho_{AA'BB'}\, \big\|\, \SEP(AA'\!:\!BB')\big) \geq D^\sep\big(\rho_{AB}\, \big\|\, \SEP(A\!:\!B)\big) + D^\sep\big(\rho_{A'B'}\, \big\|\, \SEP(A'\!:\!B')\big)
\label{strong_superadditivity}
\ee
for all (possibly correlated!) states $\rho_{AA'BB'}$ on a four-partite system $AA':BB'$.

In all cases where either~\eqref{subadditivity} or~\eqref{strong_superadditivity} is obeyed we can obtain new entanglement divergences from old ones by regularisation, i.e.\ by setting
\bb
\mathds{D}^\infty(\rho \|\SEP) \coloneqq \lim_{n\to\infty} \frac1n\, \mathds{D}\big(\rho^{\otimes n} \,\big\|\, \SEP_n\big)\, ,\qquad \mathds{D} = D,\, D_{\max},\, D^\sep\, ,
\label{regularised_entanglement_divergence}
\ee
where
\bb
\SEP_n \coloneqq \SEP(A^n\!:\! B^n)
\label{S_n}
\ee
is the set of separable states over $n$ copies of the system, and we will write $\big(D^\sep\big)^\infty = D^{\sep,\infty}$ for brevity. The existence of the limit that defines~\eqref{regularised_entanglement_divergence} is guaranteed by Fekete's lemma~\cite{Fekete1923}, because the sequence $\mathds{D}\big(\rho^n \,\big\|\, \SEP_n\big)$ is sub-additive when $\mathds{D} = D,D_{\max}$, due to~\eqref{subadditivity}, and super-additive when $\mathds{D} = D^\sep$, due to Piani's result~\eqref{strong_superadditivity}. The measure $D^\infty(\cdot \|\SEP)$, corresponding to the choice $\mathds{D}=D$ in~\eqref{regularised_entanglement_divergence}, is known as the \deff{regularised relative entropy of entanglement}, and has been extensively investigated~\cite{Christandl2012, rel-ent-sq}. The \deff{regularised Piani relative entropy of entanglement} $D^{\sep,\infty}$, instead, has been much less studied, although Li and Winter looked into a closely related version where the set of available measurements comprises all those obtainable with local operations and one-way classical communication~\cite{rel-ent-sq}.

\subsection{Quantum Stein's lemmas}

Hypothesis testing is a general task based on distinguishing two quantum states $\rho$ and $\sigma$~\cite{HAYASHI,TOMAMICHEL}. More specifically, one is given a classical description of both states and a box that contains one copy of one of the two; the goal is to make a suitable measurement so as to guess accurately. The two types of error are (I)~type-I, which consists in guessing $\sigma$ while the state was $\rho$, and (II)~type-II, which consists in guessing $\rho$ while the state was $\sigma$. In the asymmetric setting, we treat these two errors differently. For this reason, it is very useful to consider the minimal type-II error probability that is achievable under the constraint that the type-I error probability be at most $\e\in [0,1)$. This probability is given by $2^{-D_H^\e(\rho\|\sigma)}$, where the \deff{hypothesis testing relative entropy} is defined by
\bb
D_H^\e(\rho\|\sigma) \coloneqq - \log_2 \inf\left\{ \Tr E\sigma:\ E\in [0,\id],\ \Tr E\rho\geq 1-\e \right\} .
\label{DH}
\ee

The above expression rests on the implicit assumption that all tests $E\in [0,\id]$ are available to the experimentalist. However, this is not necessarily the case, and as in the discussion preceding~\eqref{measured_relative_entropy} we could consider a restricted set $\ff\in \{\sep,\ppt\}$ of measurements that are allowed on a given bipartite system $AB$. In this setting, implicitly considered in in~\cite{brandao_adversarial}, the quantity in~\eqref{DH} is modified, and becomes
\bb
D_H^{\ff\!,\,\e}(\rho\|\sigma) \coloneqq - \log_2 \inf\left\{ \Tr E\sigma:\ E\in [0,\id],\ \left(E,\id-E\right)\in \ff,\ \Tr E\rho\geq 1-\e \right\} .
\label{free_DH}
\ee
We will call the above the \deff{$\ff$-measured hypothesis testing relative entropy}. Note that when $\ff=\mathrm{ALL}$ contains all measurements we see that~\eqref{free_DH} reduces to~\eqref{DH}.

We can easily turn the above divergence into an entanglement measure by following the blueprint of~\eqref{divergence_of_entanglement}, thus setting
\bb
D_H^{\ff\!,\,\e}(\rho\|\SEP) \coloneqq&\ \inf_{\sigma\in \SEP} D_H^{\ff\!,\,\e}(\rho\|\sigma) \\
=&\ - \log_2 \inf\left\{ \sup_{\sigma \in \SEP} \Tr E\sigma:\ E\in [0,\id],\ \left(E,\id-E\right)\in \ff,\ \Tr E\rho\geq 1-\e \right\} ,
\label{min_free_DH}
\ee
where all states and operators are on a bipartite system $AB$, and the second identity holds for all $\mm\in \{\sep,\ppt,\all\}$. We will care in particular about the choice $\mm=\sep$, in which case~\eqref{min_free_DH} reduces to
\bb
\inf_{\sigma \in \SEP} D_H^{\sep\!,\,\e} (\rho\|\sigma) = - \log_2 \inf\left\{ \sup_{\sigma \in \SEP} \Tr E\sigma:\ E\in [0,\id],\ E,\id-E\in \cone(\SEP),\ \Tr E\rho\geq 1-\e \right\} .
\label{min_SEP_DH}
\ee

\begin{rem}
To exchange infimum and supremum in the second line of~\eqref{min_free_DH} we applied Sion's theorem~\cite{Sion}. To do so, it is necessary to assume that (a)~we are working with finite-dimensional systems, as we always do here, and (b)~that the set of binary measurements in $\ff$ is closed\footnote{Being bounded, it is then automatically also compact.} and convex. The sets of measurements listed in~\eqref{measurement_inclusions} are all such that the corresponding subsets of binary measurements are closed and convex, with the exception of $\locc$. Hence, once again we exclude $\locc$ from our treatment, and consider only $\mm\in \{\sep,\ppt,\all\}$.
\end{rem}

Let $\rho=\rho_{AB}$ be a generic bipartite state. The quantity
\bb
\lim_{\e\to 0^+} \liminf_{n\to\infty} \frac1n \inf_{\sigma_n\in \SEP_n} D_H^{\ff,\e} \big(\rho^{\otimes n} \,\big\|\, \sigma_n\big) \eqqcolon \stein_{\ff} \left(\rho \| \SEP\right) ,
\label{restricted_Stein_exponent}
\ee
where we used the shorthand notation $\SEP_n\coloneqq \SEP\big(A^n\!:\! B^n\big)$ for the set of separable states on $n$ copies of the system $AB$, has an operational interpretation as the Stein coefficient associated with the hypothesis testing task of discriminating $\rho_{AB}^{\otimes n}$ from separable states in $\SEP\big(A^n\!:\! B^n\big)$ using only measurements in the restricted set $\ff\in \{\sep,\ppt,\all\}$. A so-called `generalised quantum Stein's lemma'~\cite{Brandao2010, brandao_adversarial} is a statement that simplifies the calculation of the limit in~\eqref{restricted_Stein_exponent} by eliminating the external limit in $\e$. Of the two most important cases $\ff=\sep$ and $\ff=\all$, the only one for which we have such a statement is, curiously, the former~\cite{brandao_adversarial}. As for the latter, Brand\~{a}o and Plenio proposed an argument in~\cite{Brandao2010}, but this was later found to contain a fatal flaw~\cite{gap}. Although it is still \emph{conjectured} that a generalised quantum Stein's lemma might hold for $\ff=\all$, as of the writing of this manuscript it is not known whether that is the case.
This would take the form~\cite{gap}
 \bb
 \stein_{\all} \left(\rho \| \SEP\right) = \lim_{\e\to 0^+} \liminf_{n\to\infty} \frac1n \inf_{\sigma_n\in \SEP_n} D_H^\e \big(\rho^{\otimes n} \,\big\|\, \sigma_n\big) \overset{?}{=} D^\infty(\rho \|\SEP)\, .
 \label{GQSL_conjecture}
 \ee
 What is remarkable about this conjectured identity is that on the right-hand side we have the `pure' regularised relative entropy of entanglement, without any $\e$ in sight. (There is still, however, a regularisation, i.e.\ a limit $n\to\infty$, but we know that quantum information theory tends to keep these limits around.)


Fortunately, the case of importance to us here is $\ff=\sep$, for which we \emph{do} have a generalised quantum Stein's lemma. This case follows from the main result of~\cite{brandao_adversarial}, which can be stated as follows.

\begin{thm}[(Brand\~{a}o/Harrow/Lee/Peres~\cite{brandao_adversarial})] \label{BHLP_thm}
Let $\FF$ be a family of sets of states on a collection of quantum systems $A$, and let $\ff$ be a family of set of measurements on the same systems. If $\FF$ and $\ff$ are compatible, meaning that for every $\sigma \in \FF_{AA'}$, every measurement $(E_x)_{x\in \XX}\in \ff_{A'}$, and every $x\in \XX$, it holds that
\bb
\Tr \sigma_{AA'} E_x^{A'} \in \FF_{A}\, ,
\label{compatibility}
\ee
then
\bb
\stein_{\ff} \left(\rho \| \FF\right) \coloneqq \lim_{\e\to 0^+} \liminf_{n\to\infty} \frac1n \inf_{\sigma_n\in \FF_n} D_H^{\ff,\e} \big(\rho^{\otimes n} \,\big\|\, \sigma_n\big) = D^{\ff, \infty}(\rho \| \FF)\, .
\ee
\end{thm}

Choosing $\FF=\SEP$ and $\ff=\sep$ satisfies the compatibility condition~\eqref{compatibility}, as originally deduced by Piani~\cite{Piani2009}. From Theorem~\ref{BHLP_thm} we then see that
\bb
\stein_{\sep} \left(\rho \| \SEP\right) = \lim_{\e\to 0} \lim_{n\to\infty} \frac1n \inf_{\sigma_n\in \SEP_n} D_H^{\sep\!,\,\e} \big(\rho^{\otimes n} \,\big\|\, \sigma_n\big) = D^{\sep, \infty}(\rho \| \SEP)
\label{restricted_Stein_SEP}
\ee
for all bipartite states $\rho_{AB}$, which is nothing but the advertised generalised quantum Stein's lemma for the case $\ff=\sep$.

\subsection{Smoothed divergences and asymptotic equipartition property}

The hypothesis testing relative entropy in~\eqref{DH} seems to be a rather different beast from the other divergences $D,D_{\max}, D^\mm$. However, there is a way to recover $D_H^\e$ from a `smoothed' version of the max-relative entropy. Smoothing, a technique introduced and systematically studied by Renner in his PhD thesis~\cite{RennerPhD}, consists of minimising a given divergence over a small ball centred around the first state $\rho$, while keeping the second state $\sigma$ fixed. If one uses the trace norm to define such a ball, setting
\bb
B^\e_T(\rho) \coloneqq \left\{ \tilde{\rho}:\ \tilde{\rho}\geq 0,\ \Tr \tilde{\rho} = 1,\ \frac12 \|\rho - \tilde{\rho}\|_1\leq \e\right\},
\ee
for some $\e\in [0,1)$, the corresponding \deff{smoothed max-relative entropy} takes the form
\bb
D_{\max}^{\e} (\rho \| \sigma) \coloneqq \min_{\tilde{\rho}\in B_T^\e(\rho)} D_{\max}(\tilde{\rho}\|\sigma)\, ,
\label{T_smoothed_max_relative_entropy}
\ee
where the minimum exists because $(\rho,\sigma) \mapsto D_{\max}(\rho\|\sigma)$ is lower semi-continuous and 
$B_T^\e(\rho)$ is compact.

What kind of entanglement measure do we obtain by plugging the smoothed max-relative entropy into~\eqref{divergence_of_entanglement}, i.e.\ minimising them over all separable states? This would amount to defining
\bb
D_{\max}^{\e}(\rho\|\SEP) &\coloneqq \min_{\sigma\in \SEP} D_{\max}^{\e}(\rho\|\sigma) = \min_{\sigma\in \SEP,\ \tilde{\rho} \in B^\e_T(\rho)} D_{\max}(\tilde{\rho}\|\sigma)\,
\label{smoothed_max_relent_entanglement}
\ee
where once again the minimum is achieved because the max-relative entropy $(\rho,\sigma)\mapsto D_{\max}(\rho\|\sigma)$ is lower semi-continuous, and both $\SEP$ and $B_T^\e(\rho)$ are compact. Brand\~{a}o and Plenio~\cite[Proposition~IV.2]{BrandaoPlenio2} and Datta~\cite[Theorem~1]{Datta-alias} have independently shown that by regularising the above measures and subsequently taking the limit $\e\to 0$ one gets back the familiar regularised relative entropy of entanglement. 

\begin{thm}[(Brand\~{a}o--Plenio--Datta, Asymptotic equipartition property of max-relative entropy of entanglement~\cite{BrandaoPlenio2, Datta-alias})] \label{AEP_entanglement_thm}
Let $\rho=\rho_{AB}$ be a finite-dimensional bipartite state, and let $\SEP_n \coloneqq \SEP\big(A^n\!:\!B^n\big)$ be the set of separable states over $n$ copies of $AB$. Then the measure constructed in~\eqref{smoothed_max_relent_entanglement}, with the notation defined by~\eqref{max_relative_entropy} and~\eqref{T_smoothed_max_relative_entropy}, satisfies that
\bb
\lim_{\e\to 0^+} \limsup_{n\to\infty} \frac1n\, D^{\e}_{\max}\big(\rho^{\otimes n} \,\big\|\, \SEP_n\big) = \lim_{\e\to 0^+} \liminf_{n\to\infty} \frac1n\, D^{\e}_{\max}\big(\rho^{\otimes n}\, \big\|\, \SEP_n\big) = D^\infty(\rho\|\SEP)\,.
\label{AEP_entanglement}
\ee
\end{thm}

Interestingly, the smoothed max-relative entropy $D^\ve_{\max}$ satisfies a form of duality with the hypothesis testing relative entropy $D^\ve_H$, in the sense that $D^{1-\ve}_{\max}(\rho\|\sigma)$ is approximately equivalent to $D^\ve_H(\rho\|\sigma)$~\cite{Tomamichel2013, Anshu2019}.
Because of this, we can rewrite Eq.~\eqref{AEP_entanglement} as
\bb
\lim_{\e\to 1^-} \limsup_{n\to\infty} \frac1n\, D^{\e}_{H}\big(\rho^{\otimes n} \,\big\|\, \SEP_n\big) = D^\infty(\rho\|\SEP)\,,
\ee
 where we note that the error $\e$ now goes to \emph{one}. 
Another important consequence of this point of view is that the generalised quantum Stein's lemma conjecture, formulated in~\eqref{GQSL_conjecture} in terms of the hypothesis testing relative entropy, can be equivalently rephrased in terms of a max-relative entropy in the strong converse regime. Below we record this simple fact, already stated in~\cite[Eq.~(84)]{gap}.

\begin{cor}[(Reformulation of the generalised quantum Stein's lemma conjecture)] \label{GQSL_equivalent_form_cor}
For all bipartite states $\rho=\rho_{AB}$, the generalised Stein exponent under all measurements, defined by~\eqref{GQSL_conjecture}, can be alternatively expressed as
\bb
\stein_{\all} \left(\rho \| \SEP\right) &= \lim_{\e \to 0^+} \liminf_{n\to\infty} \frac1n\, D_{\max}^{1-\e} \big(\rho^{\otimes n} \,\big\|\, \SEP_n\big)\, ,
\label{GQSL_equivalent_form}
\ee
where $\SEP_n \coloneqq \SEP\big(A^n\!:\! B^n\big)$ is the set of separable states over $n$ copies of the system $AB$.
\end{cor}

This reformulation will be the starting point of our partial solution to this problem in Supplementary Note~\ref{sec:antisym_note}.

\subsection{Special states: Werner and isotropic families}\label{sec:werner}

In what follows, we will work often with a few special states defined on a bipartite system with Hilbert space $\C^d\otimes \C^d$, conventionally referred to as a `$d\times d$ bipartite system'. Here, $d\in \N_+$ is an arbitrary positive integer. First, the \deff{maximally entangled state} is given by
\bb
\Phi_d \coloneqq \ketbra{\Phi_d}\, ,\qquad \ket{\Phi_d} \coloneqq \frac{1}{\sqrt{d}} \sum_{i=1}^d \ket{ii}\, .
\label{maximally_entangled}
\ee
We will also use the shorthand notation
\bb
\tau_d\coloneqq \frac{\id - \Phi_d}{d^2-1}\, .
\ee
A generic \deff{isotropic state}~\cite[Section~V]{Horodecki1999} can be written as $p\Phi_d + (1-p)\tau_d$, for $p\in [0,1]$. 
An important feature of isotropic states is their invariance under twirling. The \deff{twirling} map in local dimension $d$ is the quantum channel
\bb
\TT(X) \coloneqq \int \dd U\ (U\otimes U^*) X (U\otimes U^*)^\dag\, ,
\label{twirling}
\ee
where $X$ is a generic operator on $\C^d\otimes \C^d$, and $\dd U$ denotes the Haar measure over the unitary group of dimension $d$. Note that $\TT$ can be realised by using local operations and shared randomness only, and is therefore in particular an LOCC~\cite{LOCC} and a non-entangling operation~\cite{BrandaoPlenio1, irreversibility}.

For subsequent use, we record the following simple computation of the generalised robustness of isotropic states.

\begin{lemma} \label{robustness_isotropic_lemma}
For all $d\in \N_+$ and $p\in [0,1]$, the generalised robustness of the corresponding isotropic state is given by
\bb
R_\SEP^g\big(p\Phi_d + (1-p)\tau_d\big) = R_\PPT^g\big(p\Phi_d + (1-p)\tau_d\big) = (dp-1)_+\, ,
\label{robustness_isotropic}
\ee
where $x_+\coloneqq \max\{x,0\}$, and $R_\PPT^g$ is defined as in~\eqref{generalised_robustness}, but replacing the set of separable states $\SEP$ with the set of PPT states $\PPT$.

In particular, an isotropic state $p\Phi_d + (1-p)\tau_d$ is separable if and only if it is PPT, if and only if $p\leq 1/d$. Hence, every separable or PPT isotropic state is a convex combination of $\tau_d$ and $\frac1d\Phi_d + \frac{d-1}{d}\tau_d$.
\end{lemma}

\begin{proof}
The fact that $\tau_d$ and $\sigma_d \coloneqq \frac1d\Phi_d + \frac{d-1}{d}\tau_d$ are separable states can be seen by observing that they can be obtained by twirling two states which are explicitly separable, $\sigma_d = \TT(\ketbra{00})$ and $\tau_d = \TT(\ketbra{01})$~\cite{Horodecki1999}. The result for $p \leq \frac1d$ thus follows immediately.

For $p > \frac{1}{d}$, we have
\begin{equation}\begin{aligned}
    p\Phi_d + (1-p)\tau_d \leq p \Phi_d + (dp - p) \tau_d = d p\, \sigma_d,
\end{aligned}\end{equation}
which gives a feasible solution to the definition of $R^g_\SEP$ in~\eqref{generalised_robustness}, yielding $R^g_\SEP( p\Phi_d + (1-p)\tau_d) \leq dp - 1$. 
To see the opposite inequality, consider that for any $\sigma \in \PPT$ such that $p\Phi_d + (1-p)\tau_d \leq \lambda \sigma$ it holds that
\begin{equation}\begin{aligned}
    d p &= d \Tr \big[ (p\Phi_d + (1-p)\tau_d) \Phi_d \big] \\
    &\leq d \lambda \Tr \sigma \Phi_d\\
    &= d \lambda \Tr \sigma^\Gamma \Phi_d^\Gamma\\
    &\eqt{(i)} \lambda \Tr \sigma^\Gamma F\\
    &\leqt{(ii)} \lambda \left\|F\right\|_\infty\\
    &= \lambda,
\end{aligned}\end{equation}
where in (i) we used the fact that $\Phi_d^\Gamma = \frac1d F$ where $F$ is the swap operator with eigenvalues $\pm 1$, and (ii) follows by the fact that $\sigma^\Gamma$ is a valid state. This means that $R^g_\PPT(p\Phi_d + (1-p)\tau_d) \geq dp - 1$, and by the inclusion $\SEP \subseteq \PPT$, $R^g_\SEP$ cannot be smaller than $R^g_\PPT$.
\end{proof}

Isotropic states can be characterised precisely as those states $\rho$ that are invariant under twirling, i.e.\ such that $\TT(\rho) = \rho$. Using this fact, it can be shown that $\TT$ admits the alternative expression
\bb
\TT(X) = \Phi_d \Tr[X\Phi_d] + \tau_d \Tr[X (\id-\Phi_d)]\, .
\label{twirling_simplified}
\ee
An elementary and well known consequence of this fact is as follows.

\begin{lemma}[{(Twirling reduction lemma)}] \label{twirling_reduction_lemma}
Let $\NN$ be a quantum channel acting on a system of dimension $d$. Then:
\begin{enumerate}[(a)]
\item there exists a measurement operator $E\in [0,\id]$ such that
\bb
\big(\TT \circ \NN\big) (X) = \Phi_d \Tr[EX] + \tau_d \Tr[(\id-E)X]\, ;
\ee
\item there exist two states $\tilde{\rho},\omega$ such that
\bb
\big(\NN\circ \TT\big) (X) = \tilde{\rho} \Tr[X\Phi_d] + \omega \Tr[X (\id-\Phi_d)]\, .
\ee
\end{enumerate}
\end{lemma}

\begin{proof}
The claims follow straightforwardly from~\eqref{twirling_simplified}. For~(a), it suffices to set $E\coloneqq \NN^\dag(\Phi_d)$, where $\NN^\dag$ is the dual map to $\NN$, defined by~\eqref{dual_map}. As for~(b), one takes $\tilde{\rho} \coloneqq \NN(\Phi_d)$ and $\omega \coloneqq \NN(\tau_d)$.
\end{proof}

Another interesting family of states on a $d\times d$ bipartite system is that of \deff{Werner states}~\cite{Werner}. These are defined as convex combinations of the two extremal Werner states, the \deff{antisymmetric state} and the \deff{symmetric state}, respectively given by
\bb
\alpha_d \coloneqq \frac{\id - F}{d(d-1)}\, ,\qquad \sigma_d \coloneqq \frac{\id + F}{d(d+1)}\, .
\label{antisymmetric_and_symmetric}
\ee
An arbitrary Werner state can thus be written as $p\alpha_d + (1-p) \sigma_d$. Also Werner states can be thought of as belonging to a symmetric family. Defining the \deff{Werner twirling} as the modified version of the map in~\eqref{twirling} given by~\cite{Werner}
\bb
\TT_W(X) \coloneqq \int \dd U\ (U\otimes U) X (U\otimes U)^\dag\, ,
\label{Werner_twirling}
\ee
one can see that Werner states are precisely those states that are invariant under the action of $\TT_W$, i.e.\ such that $\TT_W(\rho) = \rho$. This fact allows to obtain the alternative expression
\bb
\TT_W(X) = \alpha_d \Tr\Big[X\,\frac{\id - F}{2}\Big] + \sigma_d \Tr\Big[X\,\frac{\id + F}{2}\Big]\, .
\label{Werner_twirling_simplified}
\ee
As is the case for $\TT$, also $\TT_W$ can be realised with local operations and shared randomness only, and it is therefore in particular an LOCC.

\subsection{Free operations: non-entangling and dually non-entangling}

Brand\~{a}o, Plenio~\cite{BrandaoPlenio1, BrandaoPlenio2}, and Datta~\cite{Brandao-Datta} (see also~\cite{irreversibility}) have considered entanglement manipulation under the set of \deff{non-entangling operations}, defined by
\bb
\sepp(A\to B) \coloneqq \left\{ \Xi \in \cptp(A\to B):\ \Xi(\SEP)\subseteq \SEP \right\} ,
\ee
where $\cptp(A\to B)$ denotes the set of completely positive and trace-preserving linear maps $\TT(\HH_A)\to \TT(\HH_B)$ --- here, $\TT(\HH)$ denotes the space of trace class operators on the Hilbert space $\HH$. When there is no ambiguity concerning the underlying systems, or simply when we have not specified the systems in question, we will simply write $\sepp$ instead of $\sepp(A\to B)$.

The above set of quantum channels is well known to be larger and more powerful than the set of local operations assisted by classical communication (LOCC)~\cite{BrandaoPlenio2, irreversibility, gap}. In order to get a better outer approximation of the latter, following~\cite{Chitambar2020} we can define the set of \deff{dually non-entangling operations} by setting
\bb
\dne &\coloneqq \left\{ \Xi \in \cptp:\ \Xi(\SEP)\subseteq \SEP,\ \Xi^\dag(\SEP)\subseteq \cone(\SEP) \right\} \\
&\hphantom{:}= \left\{ \Xi \in \cptp:\ \Xi(\SEP)\subseteq \SEP,\ \Xi^\dag(\sep)\subseteq \sep \right\},
\ee
where $\Xi^\dag(\sep)$ is the set of measurements $(\Xi^\dag(E_x^{AB}))_x$ for $(E_x^{AB})_x \in \sep$. It is not difficult to see that we have $\rm LOCC \subseteq SEP \subseteq DNE \subset NE$, where SEP are the separable operations~\cite{Rains1997}, and indeed all of these inclusions can be strict.

For a given $\delta\geq 0$, we can also consider the set of \deff{$\boldsymbol{\big(R_\SEP^g,\delta\big)}$-approximately non-entangling operations}, or simply \deff{$\boldsymbol{\delta}$-approximately non-entangling operations}, defined by
\bb
\sepp^g_\delta \coloneqq \left\{ \Xi \in \cptp:\ R_\SEP^g\left(\Xi(\sigma)\right) \leq \delta\ \ \forall\ \sigma\in\SEP \right\} ,
\ee
where $R^g_\SEP$, defined by~\eqref{generalised_robustness}, is the generalised robustness. Analogously, the set of \deff{$\boldsymbol{\delta}$-approximately dually non-entangling operations} can be constructed as
\bb
\dne^g_\delta \coloneqq \left\{ \Xi \in \cptp:\ R_\SEP^g\left(\Xi(\sigma)\right) \leq \delta\ \ \forall\ \sigma\in\SEP,\ \Xi^\dag(\SEP)\subseteq \cone(\SEP) \right\} .
\ee
\begin{rem}
We stress the explicit dependence of the sets of approximately (dually) non-entangling operations on the choice of the measure of entanglement $R^g_\SEP$. This is motivated by a similar choice taken in~\cite{BrandaoPlenio2}, and in particular by the fact that quantities such as entanglement cost can be evaluated exactly when such operations are considered. As shown in Ref.~\cite{irreversibility}, the asymptotic properties of entanglement manipulation depend crucially on the choice of a measure here, and even a small change may have catastrophic consequences for the transformation rates. We regard the choice of $R^g_\SEP$ as a permissive choice that leads to a well-behaved asymptotic theory, but it is by no means unique.
\end{rem}

We note here that many of our results will immediately apply not only to DNE operations, but also to \emph{dually PPT-preserving operations} and their approximate variant, defined analogously as
\bb
\dpptp &\coloneqq \left\{ \Xi \in \cptp:\ \Xi(\PPT)\subseteq \PPT,\ \Xi^\dag(\PPT)\subseteq \cone(\PPT) \right\} \\
\dpptp^g_\delta &\coloneqq \left\{ \Xi \in \cptp:\ R_\PPT^g\left(\Xi(\sigma)\right) \leq \delta\ \ \forall\ \sigma\in\PPT,\ \Xi^\dag(\PPT)\subseteq \cone(\PPT) \right\} .
\ee
Such operations can be noted to be an outer approximation to the well-studied set of PPT operations~\cite{Rains2001}, defined as those maps $\Xi$ for which $\Gamma \circ \Xi \circ \Gamma \in \rm CPTP$.

\subsection{Distillable entanglement and entanglement cost}

For $\e\in [0,1)$, $\delta\geq 0$, $\FF\in \{\sepp,\, \dne\}$, and an arbitrary bipartite state $\rho_{AB}$, one defines the associated \deff{one-shot distillable entanglement} and \deff{one-shot entanglement cost} by
\begin{align}
E_{d,\,\FF^g_\delta}^{(1),\,\e} (\rho_{AB}) &\coloneqq \max\left\{m\in \N:\, \inf_{\Theta \in \FF^g_\delta\left(AB\to A_0^{m}B_0^{m}\right)} \frac12 \left\| \Theta\left( \rho_{AB} \right) - \Phi_2^{\otimes m} \right\|_1 \leq \e \right\} , \label{one_shot_distillable} \\[.5ex]
E_{c,\,\FF^g_\delta}^{(1),\,\e} (\rho_{AB}) &\coloneqq \min \left\{m\in \N:\, \inf_{\Lambda \in \FF^g_\delta\left(A_0^{m}B_0^{m}\to AB\right)} \frac12 \left\| \Lambda\left( \Phi_2^{\otimes m} \right) - \rho_{AB} \right\|_1 \leq \e \right\} . \label{one_shot_cost}
\end{align}
If $\delta=0$, we will use also the shorthand notation
\bb
E_{d,\, \FF}^{(1),\,\e} (\rho_{AB}) \coloneqq E_{d,\, \FF^g_0}^{(1),\,\e} (\rho_{AB})\, ,\qquad E_{c,\, \FF}^{(1),\,\e}(\rho_{AB}) \coloneqq E_{c,\, \FF^g_0}^{(1),\,\e} (\rho_{AB})\, .
\ee
To define the corresponding asymptotic quantities, we pick a sequence $(\delta_n)$ of non-negative numbers $\delta_n\geq 0$. The \deff{distillable entanglement} and \deff{entanglement cost} under operations in $\FF^g_{(\delta_n)}$ are thus given by
\begin{align}
E_{d,\,\FF^g_{(\delta_n)}}^\e (\rho_{AB}) &\coloneqq \liminf_{n\to\infty} \frac1n\, E_{d,\,\FF^g_{\delta_n}}^{(1),\,\e}\!\!\left(\rho_{AB}^{\otimes n}\right) , \label{distillable_A_error} \\
E_{c,\,\FF^g_{(\delta_n)}}^\e (\rho_{AB}) &\coloneqq \limsup_{n\to\infty} \frac1n\, E_{c,\,\FF^g_{\delta_n}}^{(1),\,\e} \!\!\left(\rho_{AB}^{\otimes n}\right) . \label{cost_A_error}
\end{align}
Similarly to before, we will also write
\bb
E_{d,\,\FF}^\e(\rho_{AB}) \coloneqq E_{d,\,\FF^g_{(0)}}^\e(\rho_{AB})\, ,\qquad E_{c,\,\FF}^\e(\rho_{AB}) \coloneqq E_{c,\,\FF^g_{(0)}}^\e(\rho_{AB})\, .
\label{distillable_cost_error}
\ee
In all cases, omitting the superscript $\e$ will mean that we take the (always well-defined) limit $\e\to 0^+$. More pedantically,
\begin{align}
E_{d,\, \FF^g_{(\delta_n)}}\!(\rho_{AB}) &\coloneqq \inf_{\e\in (0,1)} E_{d,\, \FF^g_{(\delta_n)}}^\e \!(\rho_{AB}) = \lim_{\e\to 0^+} E_{d,\, \FF^g_{(\delta_n)}}^\e \!(\rho_{AB})\, , \label{distillable_A} \\
E_{c,\, \FF^g_{(\delta_n)}}\!(\rho_{AB}) &\coloneqq \sup_{\e\in (0,1)} E_{c,\, \FF^g_{(\delta_n)}}^\e \!(\rho_{AB}) = \lim_{\e\to 0^+} E_{c,\, \FF^g_{(\delta_n)}}^\e \!(\rho_{AB})\, , \label{cost_A} \\
E_{d,\, \FF}(\rho_{AB}) &\coloneqq \inf_{\e\in (0,1)} E_{d,\, \FF}^\e (\rho_{AB}) = \lim_{\e\to 0^+} E_{d,\, \FF}^\e (\rho_{AB})\, , \label{distillable} \\
E_{c,\, \FF}(\rho_{AB}) &\coloneqq \sup_{\e\in (0,1)} E_{c,\, \FF}^\e (\rho_{AB}) = \lim_{\e\to 0^+} E_{c,\, \FF}^\e (\rho_{AB})\, . \label{cost}
\end{align}

Let us establish some convenient notation. For a given positive integer $d\in \N_+$ and some measurement operator $E = E_{AB} \in [0,\id]$ on a bipartite system $AB$, we define the associated channel $\Theta_{d;\, E}$, mapping states on $AB$ to states on a bipartite $d\times d$ system, by the formula
\bb
\Theta_{d;\, E} (X) = \Phi_d \Tr[EX] + \tau_d \Tr[(\id-E)X]\, .
\label{canonical_form_distillation_map}
\ee
Given two states $\tilde{\rho} = \tilde{\rho}_{AB}$ and $\omega = \omega_{AB}$ on $AB$, we also define the map $\Lambda_{d;\, \tilde{\rho},\omega}$, mapping states on a $d\times d$ bipartite system to states on $AB$, by
\bb
\Lambda_{d;\, \tilde{\rho},\omega} (X) = \tilde{\rho} \Tr[X\Phi_d] + \omega \Tr [X (\id-\Phi_d)]\, .
\label{canonical_form_dilution_map}
\ee
With the above terminology, we can then prove the following.

\begin{lemma} \label{twirling_reduction_entanglement_manipulation_lemma}
For an arbitrary bipartite state $\rho_{AB}$, a class of free operations $\FF\in\{\sepp,\dne\}$, some $\delta\geq 0$, and an error threshold $\e\in [0,1)$, the associated one-shot distillable entanglement~\eqref{one_shot_distillable} and one-shot entanglement cost~\eqref{one_shot_cost} admit the alternate expressions
\begin{align}
E_{d,\,\FF^g_\delta}^{(1),\,\e} (\rho) &= \max\left\{m\in \N:\, \exists\ E\in [0,\id]:\ \ \Theta_{2^m;\, E} \in \FF^g_\delta\left(AB\to A_0^{m}B_0^{m}\right) ,\ \ \Tr[E\rho] \geq 1-\e\right\} , \label{one_shot_distillable_simplified} \\
E_{c,\,\FF^g_\delta}^{(1),\,\e} (\rho) &= \min \left\{m\in \N:\, \exists\ \tilde{\rho},\omega:\ \ \Lambda_{2^m;\, \tilde{\rho},\omega} \in \FF^g_\delta\left(A_0^{m}B_0^{m}\to AB\right) ,\ \ \tilde{\rho}\in B_\e(\rho) \right\} . \label{one_shot_cost_simplified}
\end{align}
\end{lemma}

\begin{proof}
Follows from Lemma~\ref{twirling_reduction_lemma}, together with the observation that the maximally entangled state is invariant under twirling~\eqref{twirling}. The details are left to the reader.
\end{proof}

\section{Distillable entanglement under (approximately) DNE operations} \label{sec_distillable}

Throughout this section we will compute the distillable entanglement under dually non-entangling operations, proving that it coincides with the regularised Piani relative entropy of entanglement $D^{\sep,\infty}$ defined in Section~\ref{subsec_entanglement_measures}.

\subsection{One-shot distillable entanglement under approximately DNE operations}

We first start with an estimate of the one-shot distillable entanglement, defined by~\eqref{one_shot_distillable} for $\FF = \dne$. The following can be thought of as a generalisation of~\cite[Theorem~1]{Brandao-Datta}.


\begin{prop} \label{one_shot_distillable_dne_prop}
For all bipartite states $\rho=\rho_{AB}$, all $\e\in [0,1]$, and all $\delta\geq 0$, it holds that
\bb
\floor{D_H^{\sep\!,\,\e}(\rho\| \SEP)} \leq E_{d,\,\dne}^{(1),\,\e} (\rho_{AB}) \leq E_{d,\,\dne^g_\delta}^{(1),\,\e} (\rho_{AB}) \leq D_H^{\sep\!,\,\e}(\rho\| \SEP) + \log_2(1+\delta)+1\, ,
\label{one_shot_distillable_dne}
\ee
where, according to~\eqref{min_free_DH}, we set $D_H^{\sep\!,\,\e}(\rho\|\SEP) \coloneqq \inf_{\sigma\in \SEP(A:B)} D_H^{\sep\!,\,\e}(\rho_{AB}\| \sigma_{AB})$.
\end{prop}

\begin{proof}
Using Lemma~\ref{twirling_reduction_entanglement_manipulation_lemma}, without loss of generality we can analyse distillation maps of the form $\Theta_{2^m;\, E}$, whose action is defined by~\eqref{canonical_form_distillation_map}, and where $E\in [0,\id]$. 
Thanks to Lemma~\ref{robustness_isotropic_lemma} we see that $\Theta_{2^m;\, E}$ is $\delta$-approximately non-entangling if and only if
\bb
\max_{\sigma\in \SEP} \Tr E\sigma \leq 2^{-m} (1+\delta)\, .
\label{ne_condition_E}
\ee
The fact that it distils entanglement from $\rho$ up to error $\e$ translates into the condition
\bb
\Tr E\rho \geq 1-\e\, .
\label{error_condition_E}
\ee
Now, note that the action of the adjoint of $\Theta_{2^m;\, E}$, defined as in~\eqref{dual_map}, is given by
\bb
\Theta_{2^m;\, E}^\dag(Y) = E \Tr \left[ \Phi_2^{\otimes m} Y\right] + (\id-E) \Tr \left[ \tau_{2^m} Y\right] .
\ee
Therefore, remembering that $0\leq \Tr \Phi_2^{\otimes m} \sigma \leq 2^{-m}$ for all separable $\sigma$, we see that $\Lambda^\dag$ is non-entangling if and only if
\bb
\id-E,\ \id + 2^m E \in \cone(\SEP)\, .
\label{dne_condition_E}
\ee
Every feasible measurement for~\eqref{min_SEP_DH}, call it $(E,\id-E)$, satisfies that $E,\id-E\in \cone(\SEP)$. As this immediately implies~\eqref{dne_condition_E}, simply because 
also $\id\in \cone(\SEP)$, we see that it yields a feasible DNE entanglement distillation protocol, whose yield is the maximum $m$ such that~\eqref{ne_condition_E} holds. This proves immediately the first inequality in~\eqref{one_shot_distillable_dne}. 

The second inequality is obvious, as $\dne\subseteq \dne^g_\delta$, hence the distillable entanglement can only increase in going from $\dne$ to $\dne^g_\delta$.

Before proving the converse bound for one-shot distillation under DNE operations, i.e.\ the third inequality in~\eqref{one_shot_distillable_dne}, it is worth comparing this problem with the approach used to establish an upper bound for one-shot distillation under non-entangling operations~\cite{BrandaoPlenio2,Brandao-Datta}. There, after performing a distillation protocol, one performs hypothesis testing with the measurement $( \Phi_2^{\otimes m}, \id - \Phi_2^{\otimes m} )$ to yield a general converse bound based on $D^\e_H$. However, this measurement is not separable --- we thus need to adjust it to find a valid measurement that can be used here.

Let us consider any feasible one-shot distillation protocol $\Lambda \in \dne^g_\delta$ such that $\frac12 \|\Lambda(\rho) - \Phi_2^{\otimes m}\|_1 \leq \e$. Since $E \in \cone(\SEP) \Rightarrow \Lambda^\dagger(E) \in \cone(\SEP)$ and
\begin{equation}\begin{aligned}
	\sup_{\sigma \in \SEP} \Tr \Lambda^\dagger(E) \sigma &\leq \sup_{\sigma' \in \SEP} \Tr E (1+\delta) \sigma'
\end{aligned}\end{equation}
due to the map $\Lambda$ being $\delta$-approximately dually non-entangling, we get that
\begin{equation}\begin{aligned}
	D_H^{\sep\!,\,\e}(\rho\|\SEP) + \log_2(1+\delta)&\geq D_H^{\sep\!,\,\e}(\Lambda(\rho)\|\SEP)\\
	&= - \log_2 \inf\left\{ \sup_{\sigma \in \SEP} \Tr E\sigma:\ (E,\id-E ) \in \sep,\ \Tr E\rho\geq 1-\e \right\} .
\end{aligned}\end{equation}
Let us then pick the measurement
\begin{equation}\begin{aligned}
	(E, \id - E) = \left( \Phi_2^{\otimes m} + \frac{1}{2^m+1} \left( \id - \Phi_2^{\otimes m}\right),\, \frac{2^m}{2^m+1} \left( \id - \Phi_2^{\otimes m}\right) \right)
\end{aligned}\end{equation}
which is separable since both $\id - \Phi_2^{\otimes m}$ and $\Phi_2^{\otimes m} + (2^m-1) \frac{\id - \Phi_2^{\otimes m}}{2^{2m} - 1}$ are in $\cone(\SEP)$ by Lemma~\ref{robustness_isotropic_lemma}, and clearly satisfies
\begin{equation}\begin{aligned}
	\Tr E\Lambda(\rho) \geq \Tr \Phi_2^{\otimes m} \Lambda(\rho) = F(\Phi_2^{\otimes m}\!, \Lambda(\rho)) \geq 1-\e
\end{aligned}\end{equation}
by hypothesis. This gives
\begin{equation}\begin{aligned}
	D_H^{\sep\!,\,\e}(\rho\|\SEP) + \log_2(1+\delta) &\geq - \log_2 \left( \sup_{\sigma \in \SEP} \Tr E\sigma \right)\\
	&=  - \log_2 \left( \sup_{\sigma \in \SEP} \Tr \left(\frac{2^m}{2^m+1}\, \Phi_2^{\otimes m} + \frac{1}{2^m+1} \id \right) \sigma \right)\\
	&= - \log_2 \frac{2}{2^m+1}\\
	&\geq m - 1.
\end{aligned}\end{equation}
\end{proof}

\begin{rem}
A more direct way to show the last inequality in~\eqref{one_shot_distillable_dne}, leading also to a slight improvement in the $\delta$ correction term, is as follows. Given an optimal $\delta$-approximate DNE entanglement distillation protocol, and hence a measurement $(E,\id-E)$ such that~\eqref{ne_condition_E},~\eqref{error_condition_E}, and~\eqref{dne_condition_E} hold for $m=E_{d,\,\dne^g_\delta}^{(1),\,\e}(\rho)$, we can define
\bb
E'\coloneqq \frac{E + 2^{-m}\id}{1+2^{-m}} \in \cone(\SEP)\, .
\ee
Observe that also
\bb
\id - E' = \frac{\id-E}{1+2^{-m}} \in \cone(\SEP)\, ,
\ee
so that $(E',\id-E')\in \sep$. Furthermore, from~\eqref{error_condition_E} we surmise that
\bb
\Tr E' \rho \geq \frac{1-\e + 2^{-m}}{1+2^{-m}} \geq 1-\e\, .
\ee
Consequently, $(E',\id-E')$ is a feasible measurement for~\eqref{min_SEP_DH}. Putting all together, this implies that
\bb
D_H^{\sep\!,\,\e}(\rho\|\SEP) &\geq - \log_2 \sup_{\sigma\in \SEP} \Tr E'\sigma \\
&= - \log_2 \frac{1}{1+2^{-m}} \left( 2^{-m} + \sup_{\sigma\in \SEP} \Tr E\sigma \right) \\
&\geq - \log_2 \frac{2^{-m} (2+\delta)}{1+2^{-m}}\\
&\geq m - \log_2(2+\delta) \\
&= E_{d,\,\dne}^{(1),\,\e}(\rho) - \log_2(2+\delta)\, .
\ee
\end{rem}

\begin{rem}
Because, as shown in Lemma~\ref{robustness_isotropic_lemma}, isotropic states are separable if and only if they are PPT, the proof method applies immediately also to dually PPT-preserving operations. Explicitly, we get
\bb
\floor{D_H^{\ppt\!,\,\e}(\rho\| \PPT)} \leq E_{d,\,\dpptp}^{(1),\,\e} (\rho_{AB}) \leq E_{d,\,\dpptp^g_\delta}^{(1),\,\e} (\rho_{AB}) \leq D_H^{\ppt\!,\,\e}(\rho\| \PPT) + \log_2(1+\delta)+1\, .
\ee
\end{rem}

\subsection{Asymptotic DNE distillable entanglement} \label{subsec_ADNE_distillable}

We are now ready to prove one of our main results, the identity between the DNE distillable entanglement and the regularised Piani relative entropy of entanglement $D^{\sep,\infty}(\cdot\|\SEP)$ defined in Section~\ref{subsec_entanglement_measures}.

\begin{manualthm}{\ref{distillable_dne_thm}}
For all bipartite states $\rho=\rho_{AB}$ the distillable entanglement under dually non-entangling operations coincides with the regularised Piani relative entropy of entanglement, in formula
\begin{equation}
E_{d,\,\dne}(\rho) = D^{\sep,\infty}(\rho\|\SEP)\, .
\label{distillable_dne}
\end{equation}
In particular, all entangled states $\rho$ are DNE distillable, i.e.\ satisfy $E_{d,\,\dne}(\rho) \geq D^{\sep}(\rho\|\SEP) > 0$.

The same result is obtained if one considers instead a sequence of $\delta_n$-approximate DNE operations, as long as $\delta_n$ does not increase exponentially fast in $n$. More precisely, if $\delta_n\geq 0$ and $\delta_n = 2^{o(n)}$ then again
\begin{equation}
E_{d,\,\dne^g_{(\delta_n)}}(\rho) = D^{\sep,\infty}(\rho\|\SEP)\, .
\label{distillable_A_dne}
\end{equation}
\end{manualthm}

\begin{proof}
It suffices to write
\bb
E_{d,\, \dne}(\rho) &\eqt{(i)} \lim_{\e\to 0^+} \liminf_{n\to\infty} \frac1n\, E_{d,\,\dne}^{(1),\,\e}\big(\rho^{\otimes n}\big) \\
&\eqt{(ii)} \lim_{\e\to 0^+} \liminf_{n\to\infty} \frac1n\, D_H^{\sep\!,\,\e}(\rho^{\otimes n}\, \big\|\, \SEP_n\big) \\
&\eqt{(iii)} D^{\sep,\infty}(\rho \|\SEP)\, ,
\ee
where (i)~follows by combining the definitions~\eqref{distillable_A_error},~\eqref{distillable_cost_error}, and~\eqref{distillable}, in~(ii) we set $\SEP_n \coloneqq \SEP\big(A^n\!:\!B^n\big)$ and employed the two-fold bound in Proposition~\ref{one_shot_distillable_dne_prop}, noting that the constant additive term that makes the rightmost side of~\eqref{one_shot_distillable_dne} larger than the leftmost side gets washed out in the limit, and finally (iii)~is due to Theorem~\ref{BHLP_thm}, originally proved in~\cite{brandao_adversarial}. The fact that the DNE distillable entanglement is faithful, i.e.\ it is nonzero for all entangled states, follows from the faithfulness of the Piani relative entropy of entanglement, which, by super-additivity, implies the faithfulness of its regularised version. For details on these latter claims we refer the reader to the original work~\cite{Piani2009}. Finally, the proof of~\eqref{distillable_A_dne} is entirely analogous.
\end{proof}

It is worth remarking here that Theorem~\ref{distillable_dne_thm} provides a strong quantitative and asymptotic 
extension of~\cite[Theorem~13]{Chitambar2020}, where the distillability under dually non-entangling operations was considered.

\begin{rem}
Following the same reasoning as before, an analogous result applies also to dually PPT-preserving operations:
    \begin{equation}
E_{d,\,\dpptp}(\rho) = D^{\ppt,\infty}(\rho\|\PPT)\, ,
\end{equation}
and hence $E_{d,\,\dpptp}(\rho) > 0$ for all $\rho \notin \PPT$. This in particular gives a direct operational meaning to the PPT-measured Piani relative entropy $D^{\ppt,\infty}(\cdot\|\PPT)$.
\end{rem}

\subsection{A general upper bound on the DNE distillable entanglement} \label{subsec_generic_upper_bound}

Throughout this section we present the novel upper bound on the regularised Piani relative entropy of entanglement discussed in the main text. 


In~\cite{Xin-exact-PPT} (see also~\cite{Wang2023}) Wang and Wilde 
studied the SDP 
\bb
E_\kappa(\rho) \coloneqq \log_2 \min\left\{ \Tr S:\ -S^\Gamma \leq \rho^\Gamma \leq S^\Gamma ,\ S\geq 0 \right\} ,
\label{kappa_entanglement}
\ee
which they called \deff{$\boldsymbol{\kappa}$-entanglement}. Here, $\rho=\rho_{AB}$ is an arbitrary bipartite state. Note that the operator $S$ appearing in~\eqref{kappa_entanglement} is not only positive semi-definite, but it also has a positive partial transpose. It is, in other words, a multiple of a PPT state. We can imagine to reinforce this constraint by demanding that $S$ belong to the cone generated by separable states, denoted by $\cone(\SEP)$. We therefore introduce the \deff{modified $\boldsymbol{\kappa}$-entanglement}, obtained by adding this constraint to the program~\eqref{kappa_entanglement}. It is given by
\bb
\kk(\rho) \coloneqq&\ \log_2 \min\left\{ \Tr S:\ -S^\Gamma \leq \rho^\Gamma \leq S^\Gamma ,\ S\!\in\! \cone(\SEP) \right\} \\
=&\ \log_2 \min\left\{ \Tr S:\ -S \leq \rho^\Gamma \leq S ,\ S\!\in\! \cone(\SEP) \right\} ,
\label{modified_kappa_entanglement}
\ee
where in the second line we substituted $S\mapsto S^\Gamma$, noting that $\cone(\SEP)$ is invariant under partial transposition, i.e.\ $S\in \cone(\SEP)$ if and only if $S^\Gamma \in \cone(\SEP)$. 
The lemma below lists some elementary properties of the above function. The reason why we are interested in it is expressed by the forthcoming Proposition~\ref{upper_bound_DNE_distillable_prop}.

\begin{lemma}[(Elementary properties of $\kk$)] \label{elementary_properties_modified_kappa_lemma}
The function $\kk$ is:
\begin{enumerate}[(a)]
\item non-negative and faithful, i.e.\ such that
\bb
\kk(\rho_{AB})\geq 0\, ,\qquad \kk(\rho_{AB}) = 0 \ \ \Longleftrightarrow\ \ \rho_{AB} \in \SEP(A\!:\! B)\, ;
\ee
\item monotonically non-increasing under operations that are both PPT and non-entangling, i.e.\ such that
\bb
\NN \in \pptop(AB\to A'B') \cap \sepp(AB\to A'B')\quad \Longrightarrow\quad \kk\big(\NN(\rho_{AB})\big) \leq \kk(\rho_{AB})\, ,
\ee
for all states $\rho_{AB}$ on $AB$, which in particular includes all LOCC;
\item sub-additive, i.e.\ such that
\bb
\kk\left(\rho_{AB} \otimes \omega_{A'B'}\right) \leq \kk(\rho_{AB}) + \kk(\omega_{A'B'})
\ee
for all pairs of states $\rho_{AB}$ and $\omega_{A'B'}$.
\end{enumerate}
\end{lemma}

\begin{proof}
Non-negativity follows because for all operators $S$ such that $\rho^\Gamma \leq S$ we have that $1 = \Tr \rho^\Gamma \leq \Tr S$. As for faithfulness, assume that there exists a state $S$ such that $-S\leq \rho^\Gamma \leq S$, $S\in \cone(\SEP)$, and $\Tr S =1$. Then, first of all we have that $\rho^\Gamma\geq 0$, because otherwise, since $\Tr \rho^\Gamma =1$, the projector $P$ onto the positive part of $\rho^\Gamma$ would satisfy $1 < \Tr P\rho^\Gamma \leq \Tr PS \leq 1$, yielding a contradiction. Now, since $0\leq \rho^\Gamma\leq S$ and $1=\Tr\rho^\Gamma = \Tr S$, we have that $\rho^\Gamma = S$, so that $\rho^\Gamma$, and hence $\rho$ itself, is a separable state. This proves claim~(a).

We now move on to~(b). Let $\NN$ be a PPT and non-entangling channel, and let $S$ be the optimal operator for the program in the first line of~\eqref{modified_kappa_entanglement}, so that $\kk(\rho) = \log_2 \Tr S$, $-S^\Gamma \leq \rho^\Gamma \leq S^\Gamma$ and $S\in \cone(\SEP)$. Then $T\coloneqq \NN(S)$ satisfies that
\bb
T^\Gamma \pm \NN(\rho)^\Gamma = \left( \Gamma \circ \NN \circ \Gamma\right)\left( S^\Gamma \pm \rho^\Gamma\right) \geq 0\, ,
\ee
because $\Gamma\circ \NN\circ \Gamma$ is still completely positive, as $\NN$ is a PPT channel. Also, due to the fact that $\NN$ is non-entangling we also deduce that $T\in \cone(\SEP)$. Hence, $\kk\big(\NN(\rho)\big) \leq \log_2 \Tr T = \log_2 \Tr S = \kk(\rho)$, proving claim~(b).

Sub-additivity is also straightforward. Pick two optimal operators $S$ and $T$ for the programs on the second line of~\eqref{kappa_entanglement} written for $\rho$ and $\omega$, respectively. Therefore, on the one hand $S \pm \rho^\Gamma\geq 0$, $S\in \cone\big(\SEP(A\!:\!B)\big)$, and $\kk(\rho) = \log_2 \Tr S$, while on the other $T\pm \omega^\Gamma \geq 0$, $T\in \cone\big(\SEP(A'\!:\!B')\big)$, and $\kk(\omega) = \log_2 \Tr T$. Note that $S\otimes T\in \cone\big(\SEP(AA'\!:\!BB')\big)$, and moreover
\bb
S\otimes T \pm \rho^\Gamma \otimes \omega^\Gamma = \frac12 \left( \big(S + \rho^\Gamma\big) \otimes \big(T \pm \omega^\Gamma\big) + \big(S - \rho^\Gamma\big) \otimes \big(T \mp \omega^\Gamma\big) \right) \geq 0\, .
\ee
Thus,
\bb
\kk(\rho\otimes \omega) \leq \log_2 \Tr \left[ S\otimes T \right] = \log_2 \Tr S + \log_2 \Tr T = \kk(\rho) + \kk(\omega)\, ,
\ee
establishing~(c).
\end{proof}

The main reason why we constructed the above function $\kk$ is that it yields an easily computable upper bound on the distillable entanglement under DNE operations. Before to come to the proof of this claim, note that by Lemma~\ref{elementary_properties_modified_kappa_lemma}(c) the regularisation
\bb
\kk^\infty(\rho) \coloneqq \lim_{n\to\infty} \frac1n \kk\big(\rho^{\otimes n}\big)
\ee
is well defined for all states $\rho$, due to Fekete's lemma~\cite{Fekete1923}. From sub-additivity it also clearly follows that $\kk^\infty(\rho)\leq \kk(\rho)$.

\begin{prop}\label{upper_bound_DNE_distillable_prop}
Let $\rho=\rho_{AB}$ be a finite-dimensional bipartite quantum state. Then it holds that
\begin{equation}
E_{d,\,\dne}(\rho) = D^{\sep,\infty}(\rho \| \SEP) \leq \kk^\infty(\rho) \leq \kk(\rho)\, .
\tag{\ref{upper_bound_DNE_distillable}}
\end{equation}
\end{prop}

\begin{proof}
We write the chain of inequalities
\bb
D^{\sep}(\rho \| \SEP) &\leqt{(i)} D^{\ppt} (\rho \| \SEP) \\
&= \min_{\sigma\in \SEP} \max_{\MM\in \ppt} D\!\left(\MM(\rho) \,\big\|\, \MM(\sigma)\right) \\
&\leqt{(ii)} \min_{\sigma\in \SEP} \max_{\MM\in \ppt} D_{\max}\!\left(\MM(\rho) \,\big\|\, \MM(\sigma)\right) \\
&= \min_{\sigma\in \SEP} \max_{\MM\in \ppt} D_{\max}\!\left((\MM\circ \Gamma)\big(\rho^\Gamma\big) \,\big\|\, (\MM\circ \Gamma)\big(\sigma^\Gamma\big)\right) \\
&\leqt{(iii)} \min_{\sigma\in \SEP} D_{\max}\big(\rho^\Gamma \,\big\|\, \sigma^\Gamma \big) \\
&= \log_2 \min\left\{ \Tr S:\ \rho^\Gamma \leq S\, ,\ S\in \cone(\SEP)\right\} \\
&\leqt{(iv)} \kk(\rho)\, .
\label{upper_bound_DNE_distillable_proof}
\ee
The above derivation is justified as follows. In~(i) we expanded the set of available measurements from $\sep$ to $\ppt$, while in~(ii) we upper bounded the Umegaki relative entropy with the max-relative entropy, following~\eqref{Umegaki_max_inequality}. In~(iii) we observed that for every PPT measurement $\MM(\cdot) = \sum_x \Tr[E_x(\cdot)] \ketbra{x}$, where $E_x^\Gamma \geq 0$ for all $x$, the modified map $(\MM\circ \Gamma)(\cdot) = \sum_x \Tr\big[E_x^\Gamma (\cdot)\big] \ketbra{x}$ is still a valid quantum measurement, and hence by the data processing inequality~\eqref{data_processing_D_max} we see that
\bb
D_{\max}\!\left((\MM\circ \Gamma)\big(\rho^\Gamma\big) \,\big\|\, (\MM\circ \Gamma)\big(\sigma^\Gamma\big)\right) \leq D_{\max}\big(\rho^\Gamma \,\big\|\, \sigma^\Gamma \big)\, .
\ee
Note that $\rho^\Gamma$ may not be a state in general, however the above inequality, which rests on the general definition~\eqref{max_relative_entropy}, holds nevertheless. Finally, step~(iv) in~\eqref{upper_bound_DNE_distillable_proof} holds by direct inspection of~\eqref{modified_kappa_entanglement}.

Now, combining the above inequality~\eqref{upper_bound_DNE_distillable_proof} with Theorem~\ref{distillable_dne_thm} implies that
\bb
E_{d,\,\dne}(\rho) &= D^{\sep,\infty} (\rho \|\SEP) \\
&= \lim_{n\to\infty} \frac1n\, D^\sep \big( \rho^{\otimes n}\, \big\|\, \SEP \big) \\
&\leq \lim_{n\to\infty} \frac1n\, \kk\big( \rho^{\otimes n}\big) \\
&= \kk^\infty(\rho) \\
&\leq \kk(\rho)\, ,
\ee
completing the proof.
\end{proof}

\begin{rem}
The analogous upper bound
\begin{equation}
E_{d,\,\dpptp}(\rho) = D^{\ppt,\infty}(\rho \| \PPT) \leq E_\kappa^\infty(\rho) \leq E_\kappa(\rho)
\end{equation}
holds for dually PPT-preserving operations. Here, the original monotone $E_\kappa$ from~\cite{Xin-exact-PPT} makes a direct appearance.
\end{rem}

\section{Entanglement cost under (approximately) DNE operations} \label{sec_cost}

Throughout this section we will analyse entanglement dilution under (approximately) dually non-entangling operations. After proving a preliminary technical lemma in Section~\ref{subsec_general_lemma}, i.e.\ Lemma~\ref{one_shot_cost_key_lemma}, in Section~\ref{subsec_cost_DNE} we establish some of our main results. Intuitively speaking, the common thread of all of these results is that the entanglement cost is not significantly affected by the restriction to \emph{dually} non-entangling operations, as compared to vanilla non-entangling ones.

There are two cases to consider: first, if a small $\delta$ is allowed in the definition of approximately (dually) non-entangling operations, then Brand\~{a}o and Plenio~\cite{BrandaoPlenio2} (see also~\cite{Brandao-Datta}) showed that the corresponding asymptotic cost coincides with the regularised relative entropy of entanglement. In Corollary~\ref{one_shot_cost_A_dne_ne_cor} we show that the same is true for $\delta$-approximate DNE. On the other hand, there is the $\delta=0$ case, i.e.\ that of exactly (dually) non-entangling operations. We showed previously~\cite{irreversibility} that in this case the NE entanglement cost can be strictly larger than the regularised relative entropy of entanglement. We also left as an open problem that of calculating that cost in terms of a possibly regularised quantity. We do not solve this problem here, but we limit ourselves to showing that the same formula will also apply to DNE operations, because --- once again --- the NE entanglement cost and the DNE entanglement cost coincide (Corollary~\ref{cost_dne=ne_cor}). 

The statements of the results in this section will heavily rely on an entanglement monotone that we introduced before, namely the generalised robustness $R_\SEP^g(\rho) = \min\left\{ r\geq 0:\ \exists\ \sigma\in \SEP:\ \rho \leq (1+r)\sigma \right\}$ and the closely related max-relative entropy $D_{\max}(\rho\|\SEP) = \log_2\left(1+R_\SEP^g(\rho)\right)$. Here we will also employ the smoothed variant of the latter quantity, which we recall to be defined as
\bb
D_{\max}^{\e} (\rho \| \sigma) \coloneqq \min_{\tilde{\rho}\in B_T^\e(\rho)} D_{\max}(\tilde{\rho}\|\sigma)\, ,
\ee
where
\bb
B^\e_T(\rho) \coloneqq \left\{ \tilde{\rho}:\ \tilde{\rho}\geq 0,\ \Tr \tilde{\rho} = 1,\ \frac12 \|\rho - \tilde{\rho}\|_1\leq \e\right\}.
\label{T_ball}
\ee

\subsection{A general lemma} \label{subsec_general_lemma}

\begin{lemma} \label{one_shot_cost_key_lemma}
Let $\tilde{\rho} = \tilde{\rho}_{AB}$ be a bipartite state. For any other bipartite state $\omega=\omega_{AB}$ and some $d\in \N_+$, construct the channel 
$\Lambda_{d;\, \tilde{\rho},\omega} (X) = \tilde{\rho} \Tr[X\Phi_d] + \omega \Tr [X (\id-\Phi_d)]$ as before. Then, for all $\delta\geq 0$, we have that:
\begin{enumerate}[(a)]
\item $\Lambda_{d;\,\tilde{\rho},\omega} \in \sepp^g_\delta$ if and only if there exist $\sigma_1,\sigma_2\in \SEP$ such that
\bb
\omega\leq (1+\delta)\sigma_1\, ,\qquad \frac1d\, \tilde{\rho} + \frac{d-1}{d}\, \omega\leq (1+\delta)\sigma_2\, ;
\label{ne_condition_dilution_map}
\ee
\item $\Lambda_{d;\,\tilde{\rho},\omega} \in \dne^g_\delta$, if and only if~\eqref{ne_condition_dilution_map} is obeyed for some $\sigma_1,\sigma_2\in \SEP$, and in addition
\bb
\sup_{\sigma\in \SEP} \frac{\Tr[\sigma \tilde{\rho}]}{\Tr[\sigma \omega]} \leq d+1\, .
\label{dne_condition_dilution_map}
\ee
\end{enumerate}
Therefore,
\begin{enumerate}[(a)] \setcounter{enumi}{2}
\item if $\Lambda_{d;\,\tilde{\rho},\omega}\in \sepp^g_\delta$ then $d\geq \frac{1}{1+\delta}\left(1+R^g_\SEP(\tilde{\rho})\right)$;
\item vice versa, if 
\bb
d\geq \frac{1 + \max\big\{ R^g_\SEP(\tilde{\rho}),\, \delta^{-1}\big\}}{1 + \max\big\{ R^g_\SEP(\tilde{\rho})^{-1},\, \delta\big\}}\, ,
\ee
then there exists a state $\omega$ such that the corresponding $\Lambda_{d;\,\tilde{\rho},\omega}$ defined by~\eqref{canonical_form_dilution_map} satisfies that $\Lambda_{d;\,\tilde{\rho},\omega}\in \sepp^g_\delta$;
moreover, if 
\bb
d\geq \frac{1 + \max\big\{ 1+2R^g_\SEP(\tilde{\rho}),\, \delta^{-1}\big\}}{1 + \max\big\{ \big(1+2R^g_\SEP(\tilde{\rho}\big)^{-1},\, \delta\big\}}\, ,
\ee
then we may also have $\Lambda_{d;\,\tilde{\rho},\omega}\in \dne^g_\delta$ for a suitable choice of $\omega$;

\item finally, if for a certain $d$ and some choice of $\tilde{\rho},\omega$ the map $\Lambda_{d;\,\tilde{\rho},\omega}$ in~\eqref{canonical_form_dilution_map} satisfies that $\Lambda_{d;\,\tilde{\rho},\omega}\in \sepp^g_\delta$, then there exists another state $\omega'$ such that $\Lambda_{2d;\, \tilde{\rho},\omega'}\in \dne^g_\delta$.
\end{enumerate}
\end{lemma}

\begin{proof}
Let us proceed claim by claim.
\begin{enumerate}[(a)]
\item Since $\Lambda_{d;\,\tilde{\rho},\omega}$ is invariant under twirling from the right, i.e.\ $\Lambda_{d;\,\tilde{\rho},\omega}\circ \TT = \Lambda_{d;\,\tilde{\rho},\omega}$, and $\TT$ is a separability-preserving map, we can assume without loss of generality that the input is isotropic when deciding whether $\Lambda_{d;\,\tilde{\rho},\omega}\in \sepp^g_\delta$. Every separable isotropic state is a convex combination of $\tau_d$ and $\frac1d \Phi_d + \frac{d-1}{d}\tau_d$ due to Lemma~\ref{robustness_isotropic_lemma}, therefore thanks to the convexity of the generalised robustness we have that $\Lambda_{d;\,\tilde{\rho},\omega}\in \sepp^g_\delta$ iff
\bb
\delta &\geq \max\left\{ R_\SEP^g\left(\Lambda_{d;\,\tilde{\rho},\omega}(\tau_d)\right),\, R_\SEP^g\left(\Lambda_{d;\,\tilde{\rho},\omega}\left( \frac1d \Phi_d + \frac{d-1}{d}\tau_d \right) \right) \right\} \\
&= \max\left\{ R_\SEP^g(\omega),\, R_\SEP^g\left( \frac1d\, \tilde{\rho} + \frac{d-1}{d}\,\omega \right) \right\}
\ee
which is equivalent to the existence of $\sigma_1,\sigma_2\in \SEP$ such that~\eqref{ne_condition_dilution_map} is satisfied.
\item The additional condition that $\Lambda_{d;\,\tilde{\rho},\omega}$ needs to obey in order to be dually non-entangling is that $\Lambda_{d;\,\tilde{\rho},\omega}^\dag(\SEP)\subseteq \cone(\SEP)$. Note that $\Lambda_{d;\,\tilde{\rho},\omega}^\dag$ acts as
\bb
\Lambda_{d;\,\tilde{\rho},\omega}^\dag(Y) = \Phi_d \Tr[Y\tilde{\rho}] + (\id-\Phi_d) \Tr [Y\omega]\, .
\ee
Thus, for every $\sigma\in \SEP$ it ought to be true that
\bb
\Lambda_{d;\,\tilde{\rho},\omega}^\dag(\sigma) = \Phi_d \Tr[\sigma \tilde{\rho}] + (\id-\Phi_d) \Tr [\sigma \omega] \in \cone(\SEP)\, .
\ee
Due to Lemma~\ref{robustness_isotropic_lemma}, we know that $a\Phi_d + b(\id-\Phi_d)\in \cone(\SEP)$ if and only if $b\geq 0$ and $0\leq a\leq (d+1)b$. Hence, the additional condition to be obeyed in order for $\Lambda_{d;\,\tilde{\rho},\omega}$ to be non-entangling is seen to be precisely~\eqref{dne_condition_dilution_map}.

\item If $\Lambda_{d;\,\tilde{\rho},\omega}\in \sepp^g_\delta$, and thus~\eqref{ne_condition_dilution_map} is obeyed, then
\bb
\tilde{\rho} \leq d\left(\frac1d\, \tilde{\rho} + \frac{d-1}{d}\, \omega\right) \leq d (1+\delta)\sigma_2\, ,
\ee
implying that
\bb
R^g_\SEP(\tilde{\rho}) \leq d(1+\delta) - 1\, ,
\ee
which proves claim~(c).

\item Let $R\coloneqq R^g_\SEP(\tilde{\rho})$, so that there exists a state $\tilde{\omega}$ with the property that $\frac{1}{1+R}\, \tilde{\rho} + \frac{R}{1+R}\, \tilde{\omega} \eqqcolon \tilde{\sigma}\in \SEP$ is separable. Set $\omega \coloneqq q\tilde{\rho} + (1-q)\tilde{\omega}$ for some $q\in [0,1]$ to be specified later. For an arbitrary $d$, we then have that
\bb
\frac1d\,\tilde{\rho} + \frac{d-1}{d}\,\omega &= \left( \frac1d + q\, \frac{d-1}{d}\right) \tilde{\rho} + \frac{d-1}{d}\,(1-q)\,\tilde{\omega} \\
&\leq \max\left\{ \left( \frac1d + q\, \frac{d-1}{d}\right) (1+R),\, \frac{d-1}{d}\,(1-q)\,\frac{1+R}{R}\right\} \tilde{\sigma}\, ,
\ee
while
\bb
\omega \leq \max\left\{ q(1+R),\, (1-q)\, \frac{1+R}{R}\right\} \tilde{\sigma}\, ,
\ee
so that $\Lambda\in \sepp^g_{\delta(d,R,q)}$, with $\delta(d,R,q)$ satisfying
\bb
1+\delta(d,R,q) &= \max\left\{ \left( \frac1d + q\, \frac{d\!-\!1}{d}\right) (1\!+\!R),\, \frac{d\!-\!1}{d}\,(1-q)\,\frac{1\!+\!R}{R},\, q(1\!+\!R),\, (1-q)\, \frac{1\!+\!R}{R}\right\} \\
&= \max\left\{ \left( \frac1d + q\, \frac{d\!-\!1}{d}\right) (1\!+\!R),\, (1-q)\, \frac{1\!+\!R}{R}\right\} ,
\ee
where we ignored the second and third element in the maximum because they are manifestly smaller than the fourth and the first, respectively. Now, the optimal construction, i.e.\ the one with smallest $\delta'(d,R,q)$, is obtained by minimising the above function with respect to $q\in [0,1]$. To do this, start by observing that the two straight lines $q\mapsto \frac1d + q\, \frac{d-1}{d}$ and $q\mapsto \frac{1-q}{R}$ meet at the point
\bb
q_0 = \frac{1-\frac{R}{d}}{\frac{d-1}{d}\,R + 1}\, .
\label{meeting_point_q_0}
\ee
Note that $q_0\leq 1$. Depending on whether $q_0\geq 0$ or $q_0<0$, i.e.\ whether $d\geq R$ or $d<R$, the mimimum in $q$ is achieved either at $q=q_0$ or at $q=0$, respectively. Using this observation, after a few manipulations we obtain that
\bb
\delta'(d,R) \coloneqq \min_{q\in [0,1]} \delta(d,R,q) = \left\{\begin{array}{ll} \frac{1+R}{d}-1 & \text{if $d\leq R$,} \\[1ex] \frac{R}{(d-1)R+d} & \text{if $d\geq R$.} \end{array}\right. 
\ee
All that is left to do now is to ask ourselves: for a fixed $R$ and given some $\delta\geq 0$, what is the minimum $d\in \N^+$ such that $\delta'(d,R)\leq \delta$? (For values of $d$ larger than this minimum we also have $\delta'(d,R)\leq \delta$, as $d\mapsto \delta'(d,R)\leq \delta$ is monotonically non-increasing in $d$.) The answer is easily found to be
\bb
d'(\delta,R) \coloneqq \min\{ d\in \N^+\!\!:\ \delta'(d,R)\leq \delta\} = \ceil{\frac{1+\max\{R,1/\delta\}}{1+\max\{1/R, \delta\}}}\, .
\ee
Therefore, for 
values of $d$ larger than this quantity we have constructed a state $\omega$ such that $\Lambda_{d;\, \tilde{\rho},\omega}\in \sepp^g_\delta$, proving the claim.

\item 
We can readily adapt the construction in~(d). To do so, observe that we can ensure that~\eqref{dne_condition_dilution_map} is obeyed if we demand that $\tilde{\rho} \leq (d+1)\omega$. Since with the construction in~(d) we have that
\bb
(d+1)\omega = q(d+1) \tilde{\rho} + (1-q)(d+1)\tilde{\omega}\, ,
\ee
this additional constraint is obeyed provided that $q(d+1)\geq 1$. This limits the range of $q$ from $[0,1]$ to $\left[\frac{1}{d+1},\, 1 \right]$. We thus obtain a modified optimisation over $q$, given by
\bb
\delta''(d,R) \coloneqq \min_{q\in \left[\frac{1}{d+1},\, 1 \right]} \delta(d,R,q) = \left\{\begin{array}{ll} \frac{2(1+R)}{d+1}-1 & \text{if $d\leq 2R$,} \\[1ex] \frac{R}{(d-1)R+d} & \text{if $d\geq 2R$.} \end{array}\right.
\ee
Continuing, we see that
\bb
d''(\delta,R) \coloneqq \min\{ d\in \N^+\!\!:\ \delta''(d,R)\leq \delta\} = \ceil{\frac{1+\max\big\{1+2R,\,\delta^{-1}\big\}}{1+\max\big\{(1+2R)^{-1}, \delta\big\}}}\, ,
\ee
as claimed.

\item Finally, let $\Lambda_{d;\, \tilde{\rho},\omega}\in \sepp^g_\delta$, so that by~(a) one can find $\sigma_1,\sigma_2\in \SEP$ satisfying~\eqref{ne_condition_dilution_map}. Set $d'=2d$, and consider the choice $\omega' = \frac{1}{d'} \tilde{\rho} + \frac{d'-1}{d'}\omega$. Then
\bb
\omega' = \frac12 \left( \frac1d\, \tilde{\rho} + \frac{d-1}{d}\, \omega \right) + \frac12\, \omega \leq (1+\delta) \left( \frac12\, \sigma_2 + \frac12\, \sigma_1\right) \eqqcolon (1+\delta) \sigma'_1 ,
\ee
and moreover
\bb
\frac{1}{d'}\, \tilde{\rho} + \frac{d'-1}{d'}\, \omega' &= \left(1-\frac{1}{4d}\right) \left( \frac1d\, \tilde{\rho} + \frac{d-1}{d}\, \omega \right) + \frac{1}{4d}\, \omega \\
&\leq (1+\delta) \left(\left(1-\frac{1}{4d}\right) \sigma_2 + \frac{1}{4d}\, \sigma_1 \right) \\
&\eqqcolon (1+\delta) \sigma'_2 .
\ee
Since $\sigma'_1,\sigma'_2\in \SEP$ by convexity, by what we have shown in~(a) we know already that $\Lambda_{d';\, \tilde{\rho},\omega}\in \sepp^g_\delta$. But it also holds that
\bb
\tilde{\rho} \leq \left(d'+1\right) \left(\frac{1}{d'} \tilde{\rho} + \frac{d'-1}{d'}\omega \right) = \left(d'+1\right) \omega'\, ,
\ee
so that in particular also condition~\eqref{dne_condition_dilution_map} is met. This in turn implies that in fact $\Lambda_{d';\, \tilde{\rho},\omega}\in \dne^g_\delta$, as claimed.
\end{enumerate}
\vspace{-2ex}
\end{proof}

\begin{rem}
The above proof shows that we could have imposed Eq.~\eqref{dne_condition_dilution_map} to hold for all states $\sigma$ rather than only separable ones. What this means is that, in Lemma~\ref{one_shot_cost_key_lemma}(e)--(f), the duals of the maps $\Lambda_{d;\, \tilde{\rho},\omega}$ and $\Lambda_{d';\, \tilde{\rho},\omega'}$ can actually be chosen to be entanglement annihilating rather than simply separability preserving.
\end{rem}

\begin{rem}
Lemma~\ref{one_shot_cost_key_lemma}(d) generalises the finding by Brand\~{a}o and Plenio~\cite[proof of Proposition~IV.2]{BrandaoPlenio2}, which in our language could be rephrased as the statement that \emph{given $d\geq R(\tilde{\rho})$, there exists a state $\omega$ such that $\Lambda_{d;\, \tilde{\rho},\omega}\in \sepp^g_\delta$, with $\delta=1/R(\tilde{\rho})$.} This statement is recovered by choosing this value of $\delta$ in Lemma~\ref{one_shot_cost_key_lemma}(d).
\end{rem}

\subsection{One-shot and asymptotic entanglement cost under approximately DNE operations} \label{subsec_cost_DNE}

Using Lemma~\ref{one_shot_cost_key_lemma}(f) leads to a straightforward proof of the equivalence of the asymptotic entanglement cost under (possibly approximate) DNE and NE operations.

\begin{cor} \label{cost_dne=ne_cor}
For all bipartite states $\rho=\rho_{AB}$, all $\e\in [0,1]$, and all $\delta\geq 0$, it holds that 
\bb
E_{c,\, \sepp^g_\delta}^\e(\rho) \leq E_{c,\, \dne^g_\delta}^\e(\rho) \leq E_{c,\, \sepp^g_\delta}^\e(\rho) + 1\, .
\label{cost_dne=ne_error_delta_n}
\ee
In particular, for all sequences $(\delta_n)$ with $\delta_n\geq 0$ it holds that
\bb
E_{c,\, \sepp^g_{(\delta_n)}}^\e(\rho) = E_{c,\, \dne^g_{(\delta_n)}}^\e(\rho)\, ,
\label{cost_dne=ne_delta_n}
\ee
implying as a special case that
\bb
E_{c,\, \sepp}(\rho) = E_{c,\, \dne} (\rho)\, .
\label{cost_dne=ne}
\ee
\end{cor}

\begin{proof}
Since $\dne^g_\delta \subseteq \sepp^g_\delta$ by definition, it holds trivially that $E_{c,\, \sepp^g_\delta}^\e(\rho_{AB}) \leq E_{c,\, \dne^g_\delta}^\e(\rho_{AB})$.
We thus move on to the second inequality in~\eqref{cost_dne=ne_error_delta_n}. Let $m$ be an achievable one-shot yield of entanglement dilution under operations in $\sepp^g_\delta$, i.e.\ let it belong to the set on the right-hand side of~\eqref{one_shot_cost_simplified} in Lemma~\ref{twirling_reduction_entanglement_manipulation_lemma} for the choice $\FF=\sepp$. Then, there exist two states $\tilde{\rho}\in B^\e_T(\rho)$ and $\omega$ such that $\Lambda_{2^m;\, \tilde{\rho},\omega}\in \sepp^g_\delta$. By Lemma~\ref{one_shot_cost_key_lemma}(f), we know that there exists a modified state $\omega'$ with the property that $\Lambda_{2^{m+1};\, \tilde{\rho},\omega'}\in \dne^g_\delta$. Applying once again Lemma~\ref{twirling_reduction_entanglement_manipulation_lemma}, we conclude that $E_{c,\, \dne^g_\delta}^\e(\rho)\leq m+1$. Taking the minimum over $m$ yields the claim~\eqref{cost_dne=ne_error_delta_n}.

To prove~\eqref{cost_dne=ne_delta_n}, it suffices to apply~\eqref{cost_dne=ne_error_delta_n} for $\rho\mapsto \rho^{\otimes n}$ and $\delta \mapsto \delta_n$, divide both sides by $n$, and then take the limit $n\to\infty$. Finally,~\eqref{cost_dne=ne} follows by taking $\delta_n=0$ for all $n$.
\end{proof}

One of the main results of~\cite{BrandaoPlenio2} is that there exists a sequence $(\delta_n)_n$ with $\delta_n \tendsn{} 0$ such that 
\begin{equation}\begin{aligned}
    E_{c,\, \sepp^g_{(\delta_n)}}(\rho) = D^\infty(\rho \| \SEP).
\end{aligned}\end{equation}
As discussed in the main text, such entanglement manipulation protocols are termed asymptotically non-entangling (ANE), and we can define the \deff{asymptotically dually non-entangling} (ADNE) protocols analogously by imposing that $\delta_n \tendsn{} 0$. The following result generalises the one-shot characterisation of entanglement dilution in~\cite[Theorem~2]{Brandao-Datta} to DNE operations, showing in detail how the asymptotic convergence to $D^\infty(\rho \| \SEP)$ is recovered.

\begin{cor} \label{one_shot_cost_A_dne_ne_cor}
For all bipartite states $\rho=\rho_{AB}$, all $\e\in [0,1]$, and all $\delta > 0$, it holds that 
\bb
D_{\max}^{\e}\big(\rho \,\big\|\, \SEP\big) - \log_2
(1\!+\!\delta) &= \log_2\left(1 +\! \min_{\tilde{\rho} \in B^\e_T(\rho)} \!R^g_\SEP(\tilde{\rho}) \right) - \log_2(1\!+\!\delta) \\
&\leq E_{c,\, \sepp^g_\delta}^\e(\rho) \leq E_{c,\, \dne^g_\delta}^\e(\rho) \\
&\leq \ceil{\log_2 \left( \frac{2}{1+\delta}\,2^{D_{\max}^{\e}(\rho \| \SEP)} + \frac1\delta \right)} \leq \ceil{\log_2\left(2^{D_{\max}^{\e}(\rho \| \SEP)} + \frac{1}{2\delta} \right)} + 1\, . 
\label{one_shot_cost_A_dne_ne}
\ee
In particular, for all sequences $(\delta_n)_n$ such that $\delta_n > 0$ for all sufficiently large $n$ and moreover
\bb
- D_{\max}^{\e}\big(\rho^{\otimes n} \,\big\|\, \SEP_n \big) - o(n) \leq \log_2 \delta_n \leq o(n)
\label{delta_n_twofold_bound}
\ee
as $n\to\infty$, we have that
\bb
E_{c,\, \sepp^g_{(\delta_n)}}^\e(\rho) = E_{c,\, \dne^g_{(\delta_n)}}^\e(\rho) = \limsup_{n\to\infty} \frac1n\, D_{\max}^{\e}\big(\rho \| \SEP\big)\, .
\label{cost_error_A_dne_ne}
\ee
Therefore, 
if $\limsup_{n\to\infty} \frac1n \log_2 \delta_n \leq 0$ and $\liminf_{n\to\infty} \frac1n \log_2 \delta_n > - D^\infty(\rho\|\SEP)$ we deduce that
\bb
E_{c,\, \sepp^g_{(\delta_n)}}(\rho) = E_{c,\, \dne^g_{(\delta_n)}} (\rho) = D^\infty\big(\rho \| \SEP\big)\, .
\label{cost_A_dne_ne_1}
\ee
As a special case, it suffices to take 
$\delta_n = 2^{-n \left(D^\infty(\rho \| \SEP) \,-\, t \right)}$ for an arbitrary small $t>0$ to obtain
\bb
E_{c,\, \ane}(\rho) = E_{c,\, \adne} (\rho) = D^\infty\big(\rho \| \SEP\big)\, .
\label{cost_A_dne_ne_2}
\ee
\end{cor}

We remark that the choice of the sequence $(\delta_n)_n$ found in the original proof of the fact that $E_{c,\, \ane}(\rho) = D^\infty(\rho \| \SEP)$ in~\cite{BrandaoPlenio2} was $\delta_n = \Big(2^{D^{\e}_{\max}(\rho^{\otimes n}\|\SEP_n)}-1\Big)^{-1}$.

\begin{proof}
As we argued in the proof of Corollary~\ref{cost_dne=ne_cor}
the inequality $E_{c,\, \sepp^g_\delta}^\e(\rho) \leq E_{c,\, \dne^g_\delta}^\e(\rho)$ holds trivially by inclusion of the corresponding sets of free operations.
To establish the lower bound on the one-shot cost under non-entangling operations, i.e.\ the first inequality in~\eqref{one_shot_cost_A_dne_ne}, pick an $m\in \N$ that belongs to the set on the right-hand side of~\eqref{one_shot_cost_simplified} in Lemma~\ref{twirling_reduction_entanglement_manipulation_lemma} for $\FF=\sepp$. This means that there exist two states $\tilde{\rho}\in B^\e_T(\rho)$ and $\omega$ such that the associated dilution map $\Lambda_{2^m;\,\tilde{\rho},\omega}$ defined by~\eqref{canonical_form_dilution_map} satisfies that $\Lambda_{2^m;\,\tilde{\rho},\omega} \in \sepp^g_\delta$.
Then, from Lemma~\ref{one_shot_cost_key_lemma}(c) we obtain that
\bb
2^m \geq \frac{1}{1+\delta}\left( 1+R^g_\SEP(\tilde{\rho})\right) \geq \frac{1}{1+\delta}\left( 1+ \min_{\rho'\in B^\e_T(\rho)}R^g_\SEP(\rho')\right) ,
\ee
from which the lower bound follows by taking the logarithm of both sides.

As for the upper bound on $E_{c,\, \dne^g_\delta}^\e(\rho)$ in~\eqref{one_shot_cost_A_dne_ne}, set
\bb
m = \ceil{\log_2 \left( \frac{2}{1+\delta}\,2^{D_{\max}^{\e}(\rho \| \SEP)} + \frac1\delta \right)} = \ceil{\log_2 \left( \frac{2}{1+\delta} \left(1+\min_{\rho' \in B^\e_T(\rho)} \!R^g_\SEP(\rho')\right) + \frac1\delta \right)} ,
\ee
so that there exists $\tilde{\rho}\in B^\e_T(\rho)$ with the property that
\bb
2^m \geq \frac{2}{1+\delta}\left(1+ R^g_\SEP(\tilde{\rho})\right) + \frac{1}{\delta} \geq \frac{1 + \max\big\{ 1+2R^g_\SEP(\tilde{\rho}),\, \delta^{-1}\big\}}{1 + \max\big\{ \big(1+2R^g_\SEP(\tilde{\rho}\big)^{-1},\, \delta\big\}}\, ,
\ee
where the last inequality follows from some elementary algebraic manipulations, i.e.\ by observing that
\bb
\frac{1+\max\{x,1/\delta\}}{1+\max\{1/x,\delta\}} \leq \frac{1+\max\{x,1/\delta\}}{1+\delta} = \max\left\{ \frac{1+x}{1+\delta},\, \frac{1+1/\delta}{1+\delta}\right\} = \max\left\{ \frac{1+x}{1+\delta},\, \frac1\delta\right\} \leq \frac{1+x}{1+\delta} + \frac1\delta\, ,
\ee
where in our case $x=1+ 2 R^g_\SEP(\tilde{\rho})$.
Then, by Lemma~\ref{one_shot_cost_key_lemma}(e) we can construct a dilution map $\Lambda_{2^m;\,\tilde{\rho},\omega}$ of the form~\eqref{canonical_form_dilution_map} with the property that $\Lambda_{2^m;\,\tilde{\rho},\omega}\in \dne^g_\delta$. By~\eqref{one_shot_cost_simplified} in Lemma~\ref{twirling_reduction_entanglement_manipulation_lemma} applied for $\FF=\dne$, this proves that $E_{c,\, \dne^g_\delta}^\e(\rho)\leq m$, as claimed.

Note that~\eqref{cost_error_A_dne_ne}~follows from~\eqref{one_shot_cost_A_dne_ne} by dividing both sides of~\eqref{one_shot_cost_A_dne_ne} by $n$ and taking the $\limsup$ as $n\to\infty$. The identity~\eqref{cost_A_dne_ne_1} is a consequence of the asymptotic equipartition property
\bb
\lim_{\e\to 0^+} \limsup_{n\to\infty} \frac1n D^{\e}_{\max}(\rho \|\SEP) = D^\infty(\rho\|\SEP)\, ,
\ee
proved in~\cite{BrandaoPlenio2, Datta-alias} and reported here as Theorem~\ref{AEP_entanglement_thm}. 
Eq.~\eqref{cost_A_dne_ne_2} follows immediately, because the stated choice of $\delta_n$ satisfies~\eqref{delta_n_twofold_bound}.
\end{proof}

\section{Quantum hypothesis testing with the antisymmetric state}\label{sec:antisym_note}

Throughout this section, we study the entanglement properties of the antisymmetric state
\bb
\alpha_d = \frac{\id - F}{d(d-1)}
\label{antisymmetric_redef}
\ee
(cf.\ Sec.~\ref{sec:werner}).  
We begin in Section~\ref{subsec_GQSL_alpha} by showing that 
the generalised quantum Stein's lemma holds true for this state (Theorem~\ref{Stein_alpha_thm}). The approach we will use is an application of a general strategy sketched out in~\cite[Eq.~(84)]{gap}. This determines precisely the ultimate efficiency of entanglement testing using arbitrary measurements for the antisymmetric state. We then tackle the restriction to separable measurements in Section~\ref{subsec_Stein_alpha_SEP}, proving in Proposition~\ref{restricted_Stein_antisymmetric_prop} and Theorem~\ref{antisymmetric_gap_thm} that this restriction leads to a gap that impacts in a quantifiable way the asymptotic performance of entanglement testing.

\subsection{The generalised quantum Stein's lemma for the antisymmetric state} \label{subsec_GQSL_alpha}

\begin{thm}[(Generalised quantum Stein's lemma for the antisymmetric state)] \label{Stein_alpha_thm}
For two positive integers $d,n\in \N_+$, let $\alpha_d$ be the antisymmetric state defined by~\eqref{antisymmetric_redef}. 
Then it holds that
\bb
D_{\max}\big(\alpha_d^{\otimes n}\, \big\|\, \SEP_n\big) - \log_2\frac{1}{1-\delta} \leq D_{\max}^{\delta}\big(\alpha_d^{\otimes n}\, \big\|\, \SEP_n\big) \leq D_{\max}\big(\alpha_d^{\otimes n}\, \big\|\, \SEP_n\big)\qquad \forall\ \delta\in [0,1)\, .
\label{Stein_alpha_one_shot_D_max_version}
\ee
where $\SEP_n\coloneqq \SEP(A^n\!:\!B^n)$ stands for the set of separable states over $n$ copies of $AB$ with Hilbert space $\C^d\otimes \C^d$. Therefore,
\bb
\stein_{\all}\left(\alpha_d \| \SEP\right) &= \lim_{\e\to 0^+} \liminf_{n\to\infty} \frac1n\, D_H^\e \big(\alpha_d^{\otimes n} \,\big\|\, \SEP_n\big) \\
&= \lim_{\e \to 0^+} \liminf_{n\to\infty} \frac1n\, D_{\max}^{1-\e} \big(\alpha_d^{\otimes n} \,\big\|\, \SEP_n\big) \\
&= D_{\max}^\infty(\alpha_d \|\SEP) \\
&= D^\infty(\alpha_d \|\SEP)\, ,
\label{Stein_alpha}
\ee
where $D_{\max}^{\e}$ is defined by~\eqref{smoothed_max_relent_entanglement}.
\end{thm}

\begin{proof}
Since $\SEP_n$ is invariant under local unitaries and convex, applying the Werner twirling $\TT_W^{\otimes n}$, where $\TT_W$ is defined by~\eqref{Werner_twirling}, can only decrease $D_{\max}$. Therefore, we can choose the minimising state $\rho_n$ for $D_{\max}^{\delta,\, T}\big(\alpha_d^{\otimes n}\, \big\|\, \SEP_n\big)$ to be a fixed point of $\TT_W^{\otimes n}$. In formula,
\bb
D_{\max}^{\delta}\big(\alpha_d^{\otimes n}\, \big\|\, \SEP_n\big) = \min_{\substack{\rho_n\, \in\, B^\delta_T(\alpha_d^{\otimes n}),\\ \theta_n\in \SEP_n}} D(\rho_n \| \theta_n) = \min_{\substack{\rho_n = \TT_W^{\otimes n}(\rho_n) \,\in\, B^\delta_T(\alpha_d^{\otimes n}),\\ \theta_n\in \SEP_n}} D(\rho_n \| \theta_n)\, .
\label{Stein_alpha_proof_Werner_symmetry}
\ee
Now, the point is that we know what the fixed points of $\TT_W^{\otimes n}$ look like. As argued in~\cite{Audenaert2001} (see also~\cite{Werner-symmetry}) and as one can immediately see from~\eqref{Werner_twirling_simplified}, too, they are of the form
\bb
\rho_n = \sum_{I\subseteq [n]} P_I\, \alpha_d^{\otimes I} \otimes \sigma_d^{I^c}
\ee
where $\sigma_d$ is the symmetric state defined by~\eqref{antisymmetric_and_symmetric}, $P_I$ is an arbitrary probability distribution over the set of subsets of $[n]=\{1,\ldots, n\}$, and $I^c$ is the complementary set to $I$. Now, since $\alpha_d$ and $\sigma_d$ are orthogonal, from the condition $\rho_n \in B^\delta_T\big(\alpha_d^{\otimes n}\big)$ we see that
\bb
\delta &\geq \frac12 \left\| \alpha_d^{\otimes n} - \rho_n \right\|_1 \\
&= \frac12\left\| \left(1-P_{[n]}\right) \alpha_d^{\otimes n} - \sum_{I\subseteq [n],\, I\neq [n]} P_I\, \alpha_d^{\otimes I} \otimes \sigma_d^{I^c} \right\|_1 \\
&= \frac12\left(1-P_{[n]}\right) + \frac12 \sum_{I\subseteq [n],\, I\neq [n]} P_I \\
&= 1 - P_{[n]}\, ,
\ee
i.e.\ $P_{[n]} \geq 1-\delta$. This implies that
\bb
\rho_n \geq P_{[n]}\, \alpha_d^{\otimes n} \geq \left(1-\delta\right) \alpha_d^{\otimes n}\, .
\ee
Therefore, since $\rho_n \leq 2^{D_{\max}(\rho_n \|\SEP_n)}\, \theta_n$ for some separable $\theta_n$ by definition of $D_{\max}(\cdot\|\SEP)$, we also have that
\bb
\alpha_d^{\otimes n} \leq \frac{1}{1-\delta}\, \rho_n \leq \frac{1}{1-\delta}\, 2^{D_{\max}(\rho_n \|\SEP_n)}\, \theta_n = 2^{D_{\max}(\rho_n \|\SEP_n) \,+\, \log_2\frac{1}{1-\delta}}\, \theta_n\, ,
\ee
leading naturally to
\bb
D_{\max}\big(\alpha_d^{\otimes n}\, \big\|\, \SEP_n\big) \leq D_{\max}(\rho_n \| \SEP_n) + \log_2\frac{1}{1-\delta}\, .
\ee
Choosing $\rho_n$ as the minimiser in~\eqref{Stein_alpha_proof_Werner_symmetry} yields the first inequality in~\eqref{Stein_alpha_one_shot_D_max_version}, while the second one is trivial because the function $\delta\mapsto D^{\delta}_{\max}(\rho\|\SEP)$ is monotonically non-increasing. This proves~\eqref{Stein_alpha_one_shot_D_max_version}.

To deduce~\eqref{Stein_alpha}, start by observing that by dividing~\eqref{Stein_alpha_one_shot_D_max_version} by $n$ and taking the liminf as $n\to\infty$ one sees that $\liminf_{n\to\infty} \frac1n\, D_{\max}^{\delta,\,T}\big(\alpha_d^{\otimes n}\, \big\|\, \SEP_n\big)$ actually does not depend on $\delta\in [0,1)$. Hence,
\bb
D^\infty(\alpha_d\| \SEP) &\eqt{(i)} \lim_{\delta \to 0^+} \liminf_{n\to\infty} \frac1n\, D_{\max}^{\delta} \big(\alpha_d^{\otimes n} \,\big\|\, \SEP_n\big) \\
&\eqt{(ii)} D_{\max}^\infty(\alpha_d \| \SEP) \\
&\eqt{(iii)} \lim_{\e\to 0^+} \liminf_{n\to\infty} \frac1n\, D_{\max}^{1-\e} \big(\alpha_d^{\otimes n} \,\big\|\, \SEP_n\big) \\
&\eqt{(iv)} \stein_{\all}(\alpha_d\|\SEP)\, .
\ee
Here, (i)~holds due to the Brand\~{a}o--Plenio--Datta asymptotic equipartition property~\cite{BrandaoPlenio2, Datta-alias}, here reported as Theorem~\ref{AEP_entanglement_thm}, in~(ii) and~(iii) we employed the above observation to replace $\delta$ first with $0$ and then with $1-\e$, and finally (iv)~follows from Corollary~\ref{GQSL_equivalent_form_cor}.
\end{proof}

\subsection{Separable measurements} \label{subsec_Stein_alpha_SEP}

Let us now estimate the Stein exponent that characterises asymptotic hypothesis testing of the antisymmetric state against all separable states under separable measurements. This is given by formula~\eqref{restricted_Stein_SEP}, and can thus be computed alternatively as the regularised Piani relative entropy.

\begin{prop} \label{restricted_Stein_antisymmetric_prop}
For all $d\in \N_+$, the antisymmetric state $\alpha_d$ defined by~\eqref{antisymmetric_and_symmetric} satisfies that
\bb
\log_2 \left(1+\frac{1}{d}\right) \leq \stein_{\sep} \left(\alpha_d \| \SEP\right) = D^{\sep,\infty}(\alpha_d \| \SEP) \leq \kk(\alpha_d) = \log_2\left(1+\frac{2}{d}\right) .
\label{restricted_Stein_antisymmetric}
\ee
\end{prop}

\begin{proof}
For the lower bound we write
\bb
\stein_{\sep} \left(\alpha_d \| \SEP\right) &= D^{\sep,\infty}(\alpha_d \| \SEP) \\
&\geq D^{\sep}(\alpha_d \| \SEP) \\
&= \inf_{\theta \in \SEP} D^{\sep}(\alpha_d \| \theta) \\
&\eqt{(i)} \inf_{\theta \in \SEP} D^{\sep}\big(\alpha_d \,\big\|\, \TT_W(\theta)\big) \\
&\geqt{(ii)} \inf_{\theta \in \SEP} - \log_2 \Tr\left[(\id-P) \TT_W(\theta)\right] \\
&\eqt{(iii)} -\log_2 \max\bigg\{ \Tr\left[(\id-P)\,\sigma_d \right] ,\, \Tr\left[(\id-P)\, \frac{\alpha_d + \sigma_d}{2} \right]\bigg\} \\
&\eqt{(iv)} \log_2\left(1+\frac1d\right) .
\label{restricted_Stein_antisymmetric_proof}
\ee
Here, in~(i) we observed that due to the fact that $\TT_W$ is an LOCC, and moreover $\TT_W(\alpha_d) = \alpha_d$, we can restrict the infimum to separable states $\theta$ that are also invariant under $\TT_W$, and hence to separable Werner states. The argument to prove this latter claim is somewhat analogous to that in~\cite[Section~V.A]{Werner-symmetry}: on the one hand, $\inf_{\theta \in \SEP} D^{\sep}(\alpha_d \| \theta) \leq \inf_{\theta \in \SEP} D^{\sep}\big(\alpha_d \,\big\|\, \TT_W(\theta)\big)$ by simply restricting the infimum; on the other, since $\TT_W$ is an LOCC and hence dually non-entangling,
\bb
\inf_{\theta \in \SEP} D^{\sep}(\alpha_d \| \theta) &= \inf_{\theta \in \SEP} \sup_{\MM\in \sep} D\big(\MM(\alpha_d) \,\big\|\, \MM(\theta)\big) \\
&\geq \inf_{\theta \in \SEP} \sup_{\MM\in \sep} D\big((\MM\circ \TT_W)(\alpha_d) \,\big\|\, (\MM\circ \TT_W)(\theta)\big) \\
&= \inf_{\theta \in \SEP} \sup_{\MM\in \sep} D\big(\MM(\alpha_d) \,\big\|\, \MM\big( \TT_W(\theta)\big)\big) \\
&= \inf_{\theta \in \SEP} D^{\sep}\big(\alpha_d \,\big\|\, \TT_W(\theta)\big)\, ,
\ee
where the inequality comes from the restriction of the supremum.

Continuing with the justification of~\eqref{restricted_Stein_antisymmetric_proof}, (ii)~is obtained by choosing as an ansatz the separable measurement $\left(P,\,\id-P\right)\in \sep$, where $P\coloneqq \sum_i \ketbra{ii}$, and noting that only the second operator yields a non-zero expectation value (equal to $1$) on the antisymmetric state. Incidentally, the above ansatz is taken from~\cite{VV-dh-Chernoff} (see also~\cite{Cheng2020}). The identity in~(iii) is derived remembering that by the results in~\cite{Werner, Werner-symmetry} we have that
\bb
\TT_W(\SEP) = \co\left\{ \sigma_d,\,\frac{\alpha_d+\sigma_d}{2}\right\} .
\ee
Finally, in~(iv) we computed
\bb
\Tr \left[ (\id-P)\,\sigma_d\right] &= \frac{d-1}{d+1}\, ,\\
\Tr \left[ (\id-P)\, \frac{\alpha_d+\sigma_d}{2}\right] &= \frac12 \left(1+\frac{d-1}{d+1}\right) = \frac{d}{d+1} > \frac{d-1}{d+1}\, .
\ee
This establishes the lower bound in~\eqref{restricted_Stein_antisymmetric}.

As for the upper bound, using Proposition~\ref{upper_bound_DNE_distillable_prop} we find that
\bb
\stein_{\sep} \left(\alpha_d \| \SEP\right) &= D^{\sep,\infty}(\alpha_d \| \SEP) \leq \kk(\alpha_d) \leq \log_2 \left(1+\frac{2}{d}\right) ,
\ee
where the inequality comes from the ansatz $S=\frac1d\,\alpha_d + \left(1+\frac1d\right)\sigma_d$ plugged into the SDP on the first line of~\eqref{modified_kappa_entanglement} for $\rho = \alpha_d$. Such an ansatz is valid because
\bb
\frac1d\,\alpha_d + \left(1+\frac1d\right)\sigma_d = \frac1d (\alpha_d+\sigma_d) + \sigma_d\in \cone(\SEP)\, ,
\ee 
where we used again the fact that $\sigma_d$ and $(\alpha_d+\sigma_d)/2$ are separable states~\cite{Werner, Werner-symmetry}, and moreover
\bb
\left(\frac1d\,\alpha_d + \left(1+\frac1d\right)\sigma_d\right)^\Gamma + \alpha_d^\Gamma = \frac{2}{d(d-1)} \,(\id-\Phi) &\geq 0\, , \\
\left(\frac1d\,\alpha_d + \left(1+\frac1d\right)\sigma_d\right)^\Gamma - \alpha_d^\Gamma = \frac{2}{d}\, \Phi &\geq 0\, .
\ee
This completes the proof.
\end{proof}

\begin{manualthm}{\ref{antisymmetric_gap_thm}}
For all $d\geq 13$ we have that
\bb
\stein_\sep(\alpha_d \| \SEP) = D^{\sep,\infty}(\alpha_d\|\SEP) < \frac12 \log_2 \frac43 \leq D^\infty(\alpha_d\|\SEP) = \stein_\all(\alpha_d\|\SEP)  \, ,
\ee
and in fact $\stein_\sep(\alpha_d \| \SEP)\tends{}{d\to\infty} 0$.
\end{manualthm}

\begin{proof}
Follows by combining Theorem~\ref{Stein_alpha_thm}, Proposition~\ref{restricted_Stein_antisymmetric_prop}, and the classic result by Christandl, Schuch, and Winter~\cite[Corollary~3]{Christandl2012}, which tells us that $D^\infty(\alpha_d\|\SEP)\geq \frac12 \log_2 (4/3)$ for all $d\geq 2$.
\end{proof}

\end{document}